\documentclass[sn-mathphys]{sn-jnl}
\usepackage[utf8]{inputenc}
\usepackage{amsmath,amssymb,amsthm,amsfonts,comment}
\usepackage{mathtools}
\usepackage{color}
\usepackage{lineno}
\usepackage{enumitem}

\theoremstyle{plain}

\newtheorem{lemma}{Lemma}

\newtheorem{prop}{Proposition}
\newtheorem{approximation}{Approximation}

\newcommand{\be}{\begin{equation*}}
\newcommand{\ee}{\end{equation*}}
\newcommand{\bel}{\begin{equation}}
\newcommand{\eel}{\end{equation}}
\newcommand{\bb}[1]{\mathbb{#1}}

\DeclareMathOperator{\Pow}{Power}
\DeclareMathOperator{\argmax}{argmax}

\usepackage{subfloat}

\newcommand{\edit}[1]{#1}

\begin{document}

\frenchspacing

\title[Limits of epidemic prediction using SIR models]{Limits of epidemic prediction using SIR models}

\author*[1]{Omar Melikechi}
\email{omar.melikechi@duke.edu}

\author[2]{Alexander L. Young}

\author[1]{Tao Tang}

\author[1]{Trevor Bowman}

\author[1,3]{David Dunson}

\author[4]{James Johndrow}

\affil*[1]{\orgdiv{Department of Mathematics}, \orgname{Duke University}, \city{Durham}, \state{NC}, \country{USA}}

\affil[2]{\orgdiv{Department of Statistics}, \orgname{Harvard University}, \city{Cambridge}, \state{MA}, \country{USA}}

\affil[3]{\orgdiv{Department of Statistics}, \orgname{Duke University}, \city{Durham}, \state{NC}, \country{USA}}

\affil[4]{\orgdiv{Department of Statistics}, \orgname{University of Pennsylvania}, \city{Philadelphia}, \state{PA}, \country{USA}}

\abstract{The Susceptible-Infectious-Recovered (SIR) equations and their extensions comprise a commonly utilized set of models for understanding and predicting the course of an epidemic. In practice, it is of substantial interest to estimate the model parameters based on noisy observations early in the outbreak, well before the epidemic reaches its peak.  This allows prediction of the subsequent course of the epidemic and design of appropriate interventions. However, accurately inferring SIR model parameters in such scenarios is problematic.  
This article provides novel, theoretical insight on this issue of practical identifiability of the SIR model. Our theory provides new understanding of the inferential limits of routinely used epidemic models and provides a valuable addition to current simulate-and-check methods. We illustrate some practical implications through application to a real-world epidemic data set.}

\keywords{SIR model, epidemic prediction, parameter inference, identifiability, nonlinear dynamics, hypothesis testing}


\maketitle




\section{Introduction}


The Susceptible-Infectious-Recovered (SIR) model, first introduced in the early twentieth century, is a mathematical model describing the spread of a novel pathogen through a population \cite{Kermack1927,Ross1916,Ross1917,Ross1917a}. This model is governed by the ordinary differential equations
\begin{linenomath*}
\bel \label{sir}
\frac{ds}{dt} = -\beta i s,\quad
\frac{di}{dt} = \beta i s - \gamma i,\quad 
\frac{dr}{dt} = \gamma i.
\eel
\end{linenomath*}
According to this model, the population is divided into three groups or ``compartments,'' each of which represents a proportion of the population. The \textit{susceptible compartment}, $s$, consists of the proportion of individuals who have never been infected with the pathogen. The \textit{infected compartment}, $i$, consists of the proportion of individuals who are currently infected. The \textit{removed compartment}, $r$, consists of the proportion of individuals who have either recovered from the pathogen and are immune or have died, and are therefore removed from the population. Since the SIR model assumes all recovered individuals are permanently immune to the pathogen, the value of $r$ can be obtained from $s$ and $i$ via the identity $s+i+r = 1$. 

During the last century the SIR equations have been modified and extended to model a diverse range of epidemics including Ebola, cholera, H1N1, tuberculosis, HIV/AIDS, influenza, malaria, Dengue fever, Zika, and most recently SARS-CoV-2 \cite{Brauer2019, Coburn2009, Eisenberg2013, Khaleque2017, Lee2020, Pasquali2021, Rachah2015,   Yang2020}. In many of these examples, additional terms are added to account for pathogen specific characteristics of transmission. Additional compartments may also be added to model different subpopulations. One such example is the SEIR model, which includes a subpopulation of exposed (E) but non-infectious cases \cite{Sauer2020}. Collectively, the SIR model and its extensions and variations provide epidemiologists with a vast array of interpretable and highly expressive models to understand and predict the behavior of outbreaks. However, incorporating too many features can have subtle but important drawbacks including limited or unreliable inference of model parameters early in an epidemic. The main contribution of this article is  insight on the inferential limits of epidemics (as captured by estimated parameters) which can be obtained from noisy, real-time observations of an outbreak.

Our work is motivated by the application of SIR and related compartmental models to real-time analysis of epidemics of human disease such as the 2014 Ebola outbreak in West Africa and the ongoing SARS-CoV-2 (``coronavirus'') pandemic. In such outbreaks, the initial aim of the public health response is to extinguish the epidemic while the number of infected individuals is still small, or at least to significantly slow the rate of infection to allow time for the pathogen to be better understood and effective therapeutics or vaccines to be developed. The stay-at-home orders instituted by many countries due to SARS-CoV-2 are one recent example which has had profound global economic impacts. As such, mathematical models employed in the real-time analysis of epidemics must provide accurate inferences about properties of the epidemic -- encapsulated by model parameters -- \textit{early in the epidemic}, when only a small fraction of the population has been infected. Hereafter, we refer to estimation of unknown model parameters from observations as the \textit{inverse problem}.

As noted in a review by Hamelin et al, many disease models proposed in the literature follow a similar structure: (1) a model is proposed, (2) a subset of model parameters are inferred from the literature, and (3) the remaining parameters are fit from data using least squares or maximum likelihood estimation \cite{Hamelin2020}.  In order for these parameter estimates to be reliable the parameters must be statistically \textit{identifiable}, ruling out settings in which multiple parameter values are equally consistent with observed data. Such issues were first considered in the context of compartmental models by Bellman and Astr\"{o}m in 1970 \cite{BELLMAN1970}. Specific details relevant to the SIR model may be found in \cite{Hamelin2020}. Here we provide a brief overview of the well-posedness of the inverse problem.  

A model is structurally identifiable when
there is a single value of the parameters consistent with noise-free data observations. A comprehensive review of analytic methods for assessing structural identifiability is given in \cite{Chis2011}; alternatively, software packages such as DAISY can be used \cite{Bellu2007}. There are many examples in the literature 
 \cite{Brunel2008,Chapman2009,Daly2018,Eisenberg2013,Piazzola2020,Tuncer2016,Tuncer2018,Villaverde2018}. In particular, structural identifiability of the SIR parameters, $\beta$ and $\gamma$, is well understood with strong theoretical support. See \cite{Hamelin2020} for specific cases based on different observations of the compartments. Similar considerations arise in literature related to branching process models which are also commonly used for modeling the dynamics of an outbreak. For example, Fok and Chou establish theoretical guarantees on ascertaining the progeny and lifetime distributions for Bellman-Harris processes when one knows the extinction time or population size distributions \cite{Fok2013}. In practical applications, much less is typically known about the dynamics. Laredo et al. \cite{laredo} prove that when certain branching processes are observed only up to their $n$th generation, one can infer that the true model parameter belongs to a specific subset (which depends on $n$) of parameter space, but it is impossible to infer the exact true parameter for any finite $n$.

In practice, data observed during an epidemic tend to be very noisy, so we are far from the idealized noise-free case. \textit{Practical identifiability} is the ability to discern different parameter values based on noisy observations. Despite considerable recent attention \cite{BalsaCanto2009AnII, Balsa-Canto2008, Chis2011, Srinath2010}, far less is known about practical identifiability. Present theoretical methods rely on sensitivity analysis and the computation of the Fisher information matrix, which is analytically intractable in the SIR model and its extensions.  Instead, it is common to see Monte Carlo methods employed, wherein the model is simulated for a set value of the parameters, noise is added to the simulated observations, and a fitting procedure is conducted on the noisy data \cite{Chis2011, Hamelin2020, Lee2020, Tuncer2018}.  The fidelity of parameter estimates relative to the known values is summarized using the average relative estimation error, which is then plotted as a function of the noise intensity. 

Interestingly, the lack of practical identifiability manifests in a remarkably similar manner across multiple, different model formulations even in cases where the parameters are known to be structurally identifiable.  As the magnitude of the noise is increased, Monte Carlo parameter estimates concentrate along a curve stretched throughout parameter space indicating a functional relationship between model parameters \cite{Browning2020, Eisenberg2013, Piazzola2020, Tuncer2016, Tuncer2018}. Importantly, there are often great disparities in parameter values along this curve and hence huge uncertainty in the parameters.  See Figure \ref{fig:SIRExample} for a representative example in the specific case considered herein.

The goal of this article is to provide theoretical tools for understanding practical identifiability in the context of the SIR model.  We propose a formulation based on realistic observations early in an outbreak.  Then, using linearizations similar to those of \cite{Sauer2020}, we construct analytically tractable approximations to the SIR dynamics from which theoretical guarantees of the performance of the inverse problem are developed.  We begin by introducing the model under consideration, reemphasizing ideas discussed previously to provide overt examples of the challenges of practical identifiability.


\section{Statistical model}\label{model}


The data available to infer the parameters of an SIR model are usually noisy, biased measurements of the rate of change in the size of the susceptible compartment, discretized to unit time intervals $\Delta_t = N(s_{t-1} - s_{t})$. For simplicity, we take the time unit to be one day. Here, $N$ represents the total population size in the jurisdiction under study and $s_t$ is the size of the susceptible compartment at time $t$. The quantity $\Delta_t$ is the number of newly infected individuals between day $t-1$ and day $t$. Data on daily confirmed cases, hospitalizations, or deaths are all examples of observable data that depend on the underlying value of $\Delta_t$. Specifically, all are discrete convolutions of $\Delta_t$ of the form $p \sum_{s=0}^t \Delta_s \pi_{t-s}$, where $p$ is the probability that an infected person goes on to be diagnosed, hospitalized, or die, and $\pi_k$ is the conditional probability that a person tests positive, is hospitalized, or dies $k$ days after becoming infected given that the corresponding outcome will eventually occur. It is likely that the parameters $p$ and, to a lesser extent, $\pi$ change over the course of an epidemic. However, changing values of these parameters can only make inference more difficult, and since our main focus is on studying limitations of inference, as a starting point we assume that $p$ and $\pi$ are fixed and known. 

While the inverse problem with known initial conditions but unknown parameters $\theta=(\beta,\gamma)$ is well-posed when even a partial trajectory of $\Delta_t$ is observed, in reality we observe $\Delta_t$ corrupted with noise, and we always have to work with finitely many discrete-time observations. In epidemic modeling, unlike some other inverse problems, we do not even have control of the sampling rate and are generally stuck with at best daily monitoring data. To simplify exposition, we focus on a simple but flexible noise model in which the observed data $Y_t$ are realizations of a random variable satisfying $\bb E[Y_t] = p \Delta_t$ for some known $p \in (0,1)$. In this case, $\pi_0=1$ and $\pi_k = 0 $ for $k>0$. While our results apply to many noise models, to fix ideas we begin with Gaussian noise 
\begin{linenomath*}
\bel \label{eq:NoisySIR}
Y_t = p\Delta_t + \xi_t, \quad \xi_t \sim \mathcal{N}(0,\sigma^2_t).
\eel
\end{linenomath*}
In addition to simplifying exposition, our primary motivation for choosing Gaussian noise is to illustrate that the SIR model can, as we see shortly and explain later, be practically unidentifiable even for simple, idealized models like the one above. A secondary reason is that, despite its simplicity, \eqref{eq:NoisySIR} is not entirely unrealistic. For example, suppose any two people infected on day $t$ have the same chance of eventually testing positive, that the chance any one such person tests positive is independent of whether any other such person does, and that the average number of people who became infected on day $t$ who go on to test positive is roughly $p\Delta_t$. Then in any sufficiently large population the central limit theorem implies $Y_t$, which in this case is the number of people who become infected on day $t$ and go on to get diagnosed, is approximately normally distributed with mean $p\Delta_t$ and some variance $\sigma_t$, i.e. $Y_t$ satisfies \eqref{eq:NoisySIR}.

Initially, suppose that the variances $\sigma^2_t$ in \eqref{eq:NoisySIR} are known. A simple procedure for solving the inverse problem from data $Y_t$ is maximum likelihood. The gradient of the log-likelihood can be obtained by numerically solving an extended ODE system \cite{Gronwall1919} which allows for easy fitting via gradient-based optimization methods. It can be shown that, even when the trajectory $p\Delta_t$ is observed only at discrete time intervals and the peak of infections has not yet occurred, the maximum likelihood estimator (MLE) exists and is unique, and so the model is structurally identifiable \cite{Hamelin2020}. Problems become apparent however when one seeks to study uncertainty in the estimated parameters. Figure \ref{fig:SIRExample} gives a stark indication of the challenges. We simulate data from an SIR model with parameters $\theta = (\beta,\gamma) = (0.21, 0.07)$ and initial conditions $s_0 = 1-1/N, i_0 = 1/N$ for $N = 10^7$. These parameters were selected to roughly approximate the dynamics of the coronavirus epidemic in New York City prior to the lockdown of March 16, 2020. The trajectories $s_t,i_t$ for $0 \le t \le 120$ are shown in the left panel. By $t=80$, about 1 percent of the population has been infected, and the peak size of the infected compartment occurs around $t=120$. The right panel of Figure \ref{fig:SIRExample} is obtained by repeatedly simulating data from \eqref{eq:NoisySIR} using the trajectory in the left panel, with $p=1$ and $\sigma^2_t = 100 N$ chosen for illustrative purposes. Other potentially more realistic values of $p$ and $\sigma_t$ are considered later in the text; see for example Table \ref{tab:NYC_params_by_p} in Section \ref{NYC} and Cases 1 and 2 in Section \ref{testing}. For each replicate simulation, the model is fit by maximum likelihood. The resulting estimates of $\hat \theta$ are shown in Figure \ref{fig:SIRExample}, which plots $\hat \beta$ against $\hat \gamma$. These are samples from the sampling distribution of the maximum likelihood estimator for these parameters. The estimates exhibit very tight concentration along a line of slope $1$. The variation in $\hat R_0 = \hat \beta/\hat \gamma$ observed for these values is large, ranging from $1.88$ to $5.01$. This high degree of uncertainty occurs despite the fact that we have observed data up through the time when over half the population has been infected. 

\begin{figure}
    \centering
    \begin{tabular}{cc}
    \includegraphics[width=0.5\textwidth]{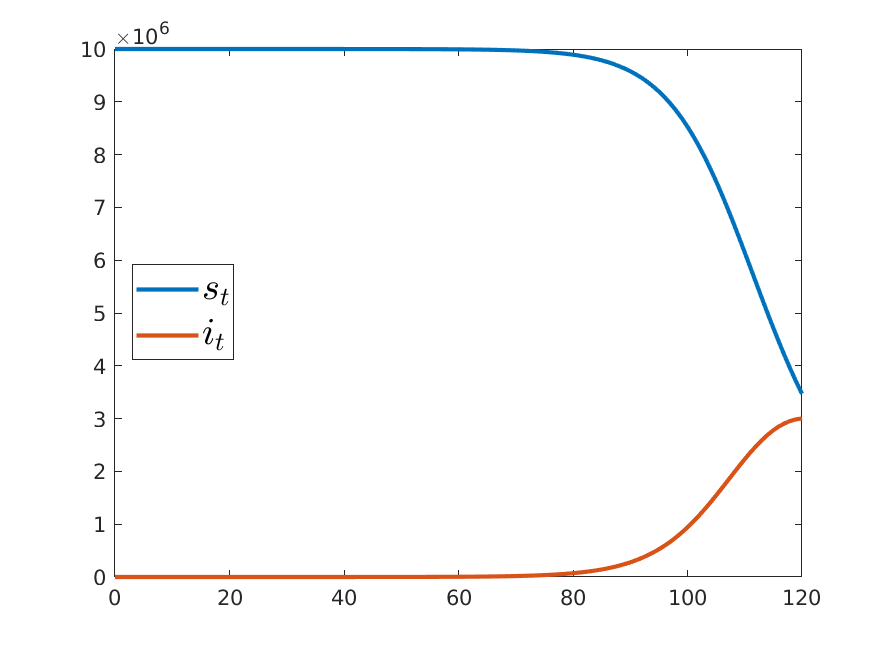} & \includegraphics[width=0.5\textwidth]{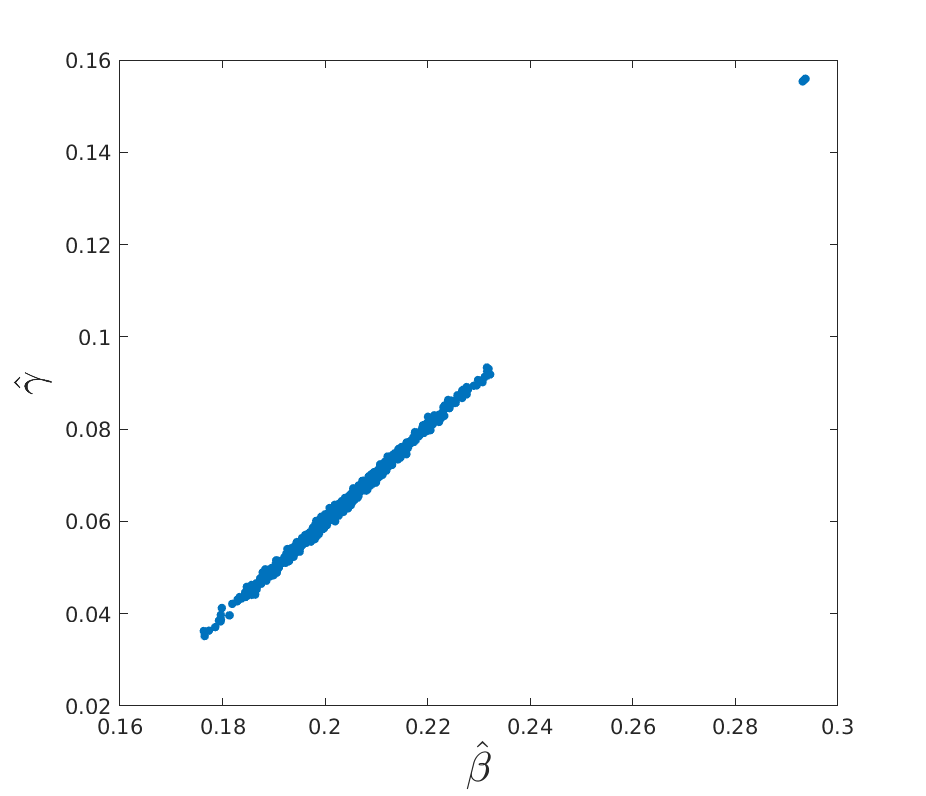} \\
    \end{tabular}
    \caption{The trajectory of the SIR model used in the simulation (left). Plots of $\hat \beta$ vs $\hat \gamma$ from 1000 realizations from the sampling distribution of their MLE (right). }
    \label{fig:SIRExample}
\end{figure}

The linear shape of the plot in Figure \ref{fig:SIRExample} suggests a practical identifiability problem in this model. That is, while the MLE exists and is unique, the curvature of the log-likelihood in the neighborhood of the MLE is very small in the direction where $\hat \beta, \hat \gamma$ lie along a line. We are not the first to notice this phenomenon. Previous works include \cite{Chis2011, Hamelin2020, Lee2020, Tuncer2018}, which experience qualitatively similar issues despite notable differences in the formulation of the likelihood in those settings.

While various empirical studies exist, our main contribution is a theoretical analysis of this phenomenon and the resulting limitations for solving the inverse problem from noisy observations. We take a two-step approach to the analysis. First, we characterize sensitivity of trajectories $\Delta_t$ to perturbations of the parameters $\theta$, and show that perturbations of $\theta$ in the directions $\pi/4$ and $5\pi/4$ (equivalently, along the line of slope $1$ through $\theta$), closely approximate the smallest variation in the trajectory $s_t$ among all perturbations for which $\|\theta_\epsilon-\theta\| = \epsilon$. We then give a computable approximate lower bound on $\inf_{\theta_\epsilon : \|\theta - \theta_\epsilon\| = \epsilon} \lvert s_t(\theta_\epsilon)-s_t(\theta)\rvert$ for times $t$ prior to the peak infection time. Taken together, these results provide an explanation for the phenomenon in Figure \ref{fig:SIRExample}. 

In the second part of the analysis, we relate the problem of uncertainty quantification to hypothesis tests of the form 
\begin{linenomath*}
\be
H_0 : \theta = \theta_0 \quad\text{ vs. }\quad H_1 : \theta = \theta_\epsilon
\ee
\end{linenomath*}
for $\|\theta_\epsilon - \theta_0\| = \epsilon$. We use the result of the first part of our analysis to approximate the type II error of the test, which in turn allows for both theoretical and empirical analysis of the limits of epidemic prediction using SIR models.

\section{Results}



\subsection{Perturbation bound for SIR trajectories}\label{bounds}


Informally, the phenomenon in Figure \ref{fig:SIRExample} is a manifestation of the fact that very different values of $\theta$ can lead to SIR model trajectories that are very close. To formalize this, let $\varphi_t(x_0,\theta)$ be the $(s,i)$-trajectory of the SIR model starting from $x_0=(s_0,i_0)$ with parameters $\theta=(\beta,\gamma)$. To aid the reader, all relevant notation is summarized in Table \ref{tab:notation}. We also remark that the analysis in this subsection and its associated appendices, Appendices \ref{sec:prop1} and \ref{sec:error}, applies directly to the deterministic SIR system \eqref{sir}. In particular, it is independent of our choice of statistical model, which will not become relevant until our discussion of hypothesis testing in Section \ref{testing}.

\begin{table}
    \centering
    \begin{tabular}{c|p{0.6\textwidth}}
    \hline
    \textbf{Notation}  &  \textbf{Description} \\
    \hline
    $x_0 = (s_0,i_0)$ & Shorthand for initial conditions with $s_0+i_0=1$ \\
    \hline
    $\theta = (\beta,\gamma)$   & Shorthand for the parameters of the SIR model \\
    \hline
    $R_0 = \beta/\gamma$ & The reproductive number \\
    \hline
    $\delta = \beta - \gamma$ & An important combination of the model parameters appearing in later analysis \\
    \hline
    $\theta_\epsilon$   & Perturbation of $\theta$ such that $\lVert\theta_\epsilon - \theta\rVert = \epsilon$ \\
    \hline
    $\theta_\epsilon(\omega)$   & Perturbation of $\theta$ in the direction $\omega \in [0,2\pi)$ such that $\lVert\theta_\epsilon(\omega) - \theta\rVert = \epsilon$ \\
    \hline
    $\varphi_t(x_0,\theta) = \left(s_t(x_0,\theta), i_t(x_0,\theta)\right)$   & Solution of the SIR equation with initial condition $x_0$ and parameter $\theta$\\
    \hline
    $Y_{1:T}$ & Observed data on days 1 through $T$ \\
    \hline
    $L(Y_{1:T}\vert \theta)$ & Likelihood of $\theta$ given observed data
    \end{tabular}
    \vspace{\baselineskip}
    \caption{Summary of notation used throughout this article.}
    \label{tab:notation}
\end{table}

For $\epsilon>0$, let $S_\epsilon(\theta)$ denote the circle of radius $\epsilon$ about $\theta$. That is,
\begin{linenomath*}
\be
	S_\epsilon(\theta) = \{\theta_\epsilon(\omega):\omega\in[0,2\pi)\}
\ee
\end{linenomath*}
where $\theta_\epsilon(\omega)=\theta+\epsilon(\cos(\omega),\sin(\omega))$. We set $\delta= \beta-\gamma$ and assume throughout that $\delta>0$; if not, then the \textit{reproductive number} $R_0=\beta/\gamma$ is at most 1 and the epidemic does not grow even at time 0. Similarly, we assume $\epsilon<\delta$. This ensures $R_0$ values of the perturbed parameters $\theta_\epsilon(\omega)=(\beta+\epsilon\cos(\omega),\gamma+\epsilon\sin(\omega))$ are also strictly greater than 1 
\begin{linenomath*}
\be
    \beta+\epsilon\cos(\omega) - \gamma-\epsilon\sin(\omega) = \delta + \epsilon(\cos(\omega)-\sin(\omega))
    	\geq \delta-\epsilon
    	>0
\ee
\end{linenomath*}
and so $R_0(\epsilon,\omega)=(\beta+\epsilon\cos(\omega))/(\gamma+\epsilon\sin(\omega))>1$ for every $\omega$. Finally, for any fixed initial condition $x_0$ and parameter $\theta$ we define the \textit{peak time}, denoted $t_*$, to be the deterministic time at which the number of infected individuals $i_t(x_0,\theta)$ is greatest; that is, $t_*=\argmax\{i_t(x_0,\theta):t\geq 0\}$. Since $di/dt=0$ if and only if $i=0$ or $s=1/R_0$, it is follows that $t_*$ exists and is unique whenever $R_0>1$. With this notation, the main result of this subsection is the following proposition.
\begin{approximation}\label{prop1} 
Let $\lVert\cdot\rVert$ denote the Euclidean norm on $\mathbb{R}^2$ and let $t_*$ be the time of peak infection corresponding to $\theta$. Then for all $t\in [0,0.8t_*)$,
\begin{linenomath*}
\bel\label{eq:Perturbation}
	\frac{\epsilon}{\delta\sqrt{2}}\big(e^{\delta t}-1\big)i_0 \approx \inf_{\omega\in[0,2\pi)} \big\lVert \varphi_t\big(x_0,\theta_\epsilon(\omega)\big) - \varphi_t(x_0,\theta)\big\rVert. 
\eel
\end{linenomath*}
Furthermore the infimum is approximately achieved when $\omega=\pi/4$ or $5\pi/4$.
\end{approximation}
The derivation of \eqref{eq:Perturbation} is in Appendix \ref{sec:prop1}. \edit{Approximation} \ref{prop1} says for any perturbation $\theta_\epsilon(\omega)$ of $\theta$, the distance between the perturbed trajectory $\varphi_t(x_0,\theta_\epsilon(\omega))$ and true trajectory $\varphi_t(x_0,\theta)$ is approximately bounded below by the left side of \eqref{eq:Perturbation} for all times $t$ up to roughly $80\%$ of $t_*$. The ``$\approx$" in \eqref{eq:Perturbation} indicates the bound is subject to error. Specifically, our derivation of \edit{Approximation} \ref{prop1} involves two approximations: First, we approximate the SIR model by a differential equation \eqref{approxODE} whose solution $\widetilde\varphi_t$ is given by \eqref{approxsol}. Second, we use first-order Taylor expansions to approximate perturbations of $\widetilde\varphi_t$ resulting from perturbations in parameter space. Despite these approximations, numerical analysis of the error given in Appendix \ref{sec:error} indicates \eqref{eq:Perturbation} holds for a wide range of parameter values and population sizes; see Figure \ref{fig2} below and Figure \ref{fig:logerror} in Appendix \ref{sec:error}. This numerical analysis also motivates our choice of $80\%$ of the peak time as a cutoff, though this cutoff can be extended to $85\%$ or even $90\%$ for larger populations and certain parameter values\edit{; see Table \ref{tab:percentpeak}.} \edit{To complement the numerical results of Appendix \ref{sec:error}, we give a theoretical upper bound on the error in Appendix \ref{sec:theory}. The theoretical result is more mathematically rigorous than the numerical one; however, it is significantly less precise than the control on error obtained in Appendix \ref{sec:error}. We therefore use results from the numerical analysis, e.g. the 80\% threshold, of Appendix \ref{sec:error} rather than the theoretical analysis of Appendix \ref{sec:theory} for the remainder of this paper.}

\begin{subfigures}\label{fig2}
\begin{figure}[!ht]
\centering
\includegraphics[width=1\textwidth, height=.5\textheight]{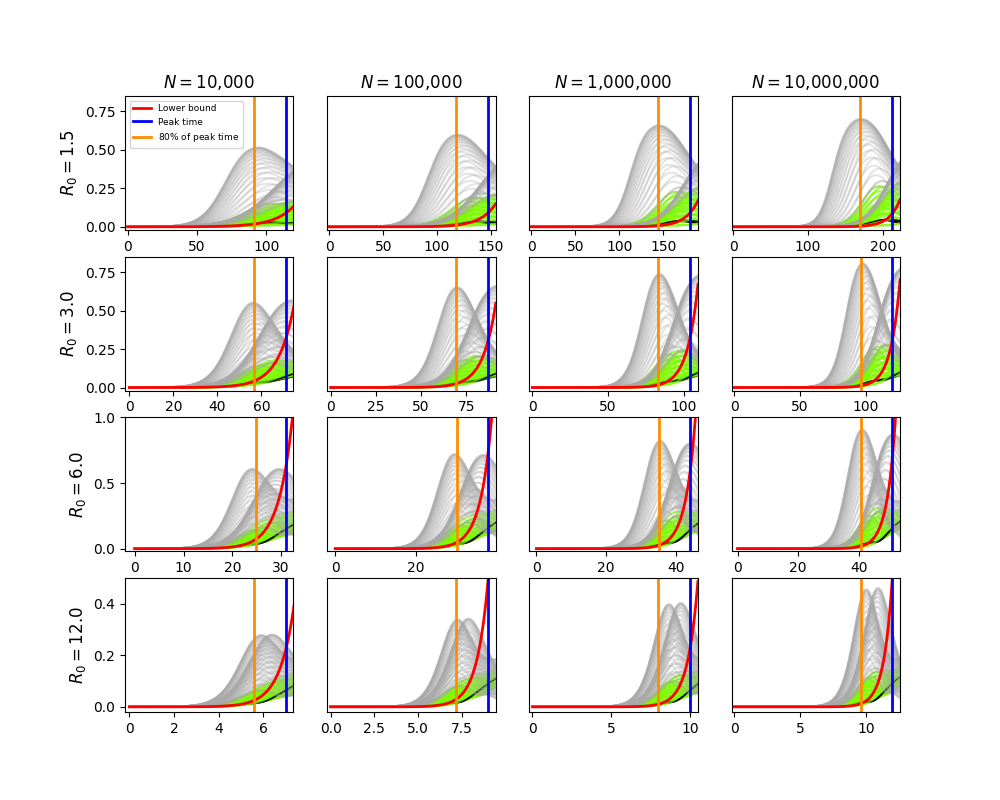}
\caption{\label{first}\textit{Distance between perturbed and true trajectories for different parameter values and population sizes}. In each graph the horizontal axis is the number of days since the start of the epidemic and the vertical axis is the distance $\lVert \varphi_t(x,\theta_\epsilon(\omega))-\varphi_t(x,\theta)\rVert$ between a perturbed trajectory and the true trajectory at time $t$. The gray, green, and black curves correspond to 90 perturbed trajectories, one for  each of 90 equally spaced angles $\omega$ in $[0,2\pi)$. The black curves correspond to the angles $\pi/4$ and $5\pi/4$. The green curves correspond to the remaining angles in the intervals $[\pi/4-\pi/12,\pi/4+\pi/12)$ and $[5\pi/4-\pi/12,5\pi/4+\pi/12)$, i.e. in intervals of width $\pi/6$ centered at $\pi/4$ and $5\pi/4$, respectively. The gray curves correspond to those angles in $[0,2\pi)$ outside these two intervals. Note the distances corresponding to angles close to $\pi/4$ and $5\pi/4$ (the green and black curves) are smaller than those distances corresponding to angles farther away from $\pi/4$ and $5\pi/4$ (the gray curves), which supports the claim that the inverse problem is least practically identifiable for parameter perturbations approximately along a line of slope 1. The approximate lower bound of \edit{Approximation} \ref{prop1} is in red. The peak time of the trajectory corresponding to $\theta$ is indicated by the vertical blue line, and 80\% of it by the vertical orange line. The first through fourth columns have population sizes $10^4, 10^5, 10^6$, and $10^7$, respectively, with only one initial infection in each case. The perturbation sizes for the first through fourth rows are $\epsilon=.03, .03, .06$, and $.1$, respectively. The SIR paramaters for the first through fourth rows are $(\beta,\gamma)=(.21,.14), (.21,.07), (.42,.07)$, and $(1.68,.14)$, which give respective $R_0$ values of 1.5, 3, 6, and 12. Note the approximate lower bound holds roughly up to 80\% of the peak time in all cases despite the wide range of parameters. Finally, we remark that the two seemingly ``distinct" classes of gray curves in each plot correspond to different subsets of the 90 distinct angles. This as well as the multimodality of certain curves (which becomes more apparent when our graphs are extended further beyond the peak time) are consequences of the nonlinearity of the SIR model and are not directly relevant to our analysis.} \label{fig2a}
\end{figure}
\begin{figure}[!ht]
\includegraphics[width=1\textwidth, height=.5\textheight]{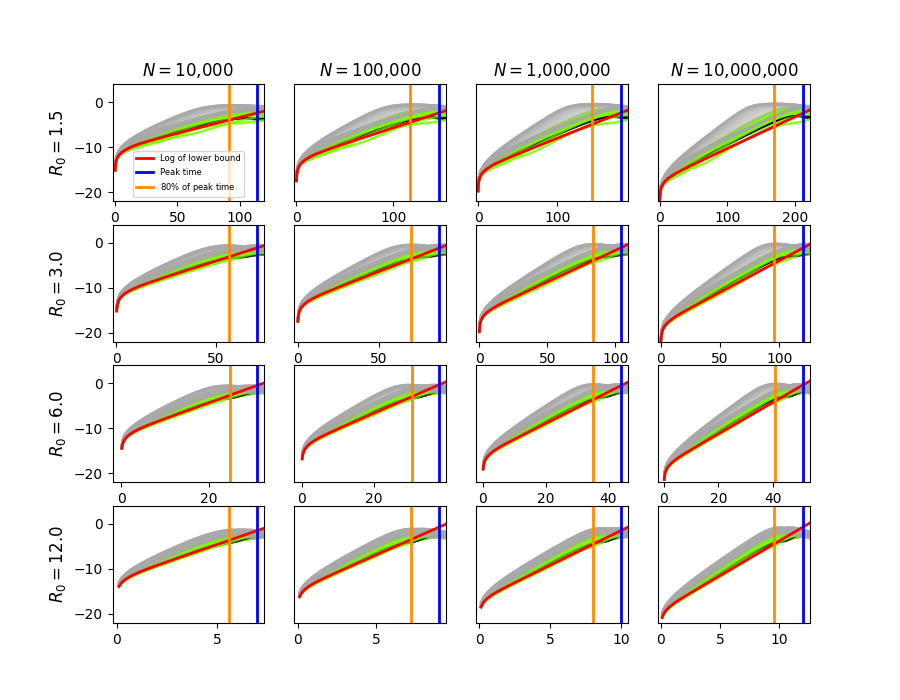}
\caption{\label{second}\textit{Logarithm of distance between perturbed and true trajectories for different parameter values and population sizes}. Everything is the same as in Figure \ref{fig2a} except now we plot $\log\lVert \varphi_t(x,\theta_\epsilon(\omega))-\varphi_t(x,\theta)\rVert$ instead of $\lVert \varphi_t(x,\theta_\epsilon(\omega))-\varphi_t(x,\theta)\rVert$. This gives a better view of the approximate lower bound early in the epidemic. Note the vertical axis is now a log scale.}
\end{figure}
\end{subfigures}

\edit{Approximation} \ref{prop1} successfully predicts the directions in parameter space, namely $\omega=\pi/4$ and $5\pi/4$ (equivalently, along the line of slope 1 through $\theta$), corresponding to the most uncertainty about parameters even when data are observed up to the peak time, as in Figure \ref{fig:SIRExample}. In other words, the inverse problem of determining $\theta$ from data is least practically identifiable when distinguishing between $\theta$ and parameter values lying approximately on the line of slope 1 through $\theta$. Furthermore, the approximate lower bound \eqref{eq:Perturbation} quantifies the extent to which the inverse problem will not be practically identifiable which, as we discuss in the next subsection, is necessary for meaningful hypothesis testing. Finally, we find that the lower bound in \eqref{eq:Perturbation} approximately holds for the $s$ trajectory alone. That is, if $s_t(x_0,\theta_\epsilon(\omega))$ and $s_t(x_0,\theta)$ are the $s$ trajectories corresponding to $\theta_\epsilon(\omega)$ and $\theta$, respectively, then
\begin{linenomath*}
\bel\label{eq:sbound}
   \frac{\epsilon}{\delta\sqrt{2}}\big(e^{\delta t}-1\big)i_0 \approx \inf_{\omega\in[0,2\pi)}\big\lvert s_t\big(x_0,\theta_\epsilon(\omega)\big) - s_t(x_0,\theta)\big\rvert,
\eel
\end{linenomath*}
and the infimum is again achieved when $\omega=\pi/4$ and $5\pi/4$. The intuition behind \eqref{eq:sbound} is that the $s$ compartment is substantially larger than the $i$ compartment early in an epidemic and therefore contributes significantly more to $\lVert\varphi^\epsilon_t-\varphi_t\rVert$ than $i$. This observation will be used for the hypothesis testing in Section \ref{testing} since our statistical model depends crucially on $\Delta_t=N(s_{t-1}-s_{t})$, which in turn depends only on $s$ rather than on $s$ and $i$ together. The approximation error implicit in the $\approx$ symbol in \eqref{eq:Perturbation} and \eqref{eq:sbound} is the one quantity we do not have rigorous control over; see Appendix \ref{sec:error} for details.


\subsection{Hypothesis testing for the inverse problem}\label{testing}


In this subsection we revisit the inverse problem in light of the perturbation bounds \eqref{eq:Perturbation} and \eqref{eq:sbound}. For context and to motivate the main result of this subsection, namely \edit{Approximation} \ref{prop2} and its subsequent discussion, we first give a brief overview of simple hypothesis testing and the Neyman-Pearson Lemma.

Suppose we observe data $Y$ taking values in a space $\mathcal{Y}$ and that these data are drawn from an unknown probability distribution belonging to a parametrized family of probability distributions $\{\mathbb{P}_\theta\}$. Given two parameters $\theta_0$ and $\theta_1$, a natural question is whether the observed data came from $\mathbb{P}_{\theta_0}$ or $\mathbb{P}_{\theta_1}$. This is a \textit{simple hypothesis test}, denoted by
\begin{linenomath*}
\bel\label{eq:hypothesis}
H_0 : \theta = \theta_0 \quad\text{ vs. }\quad H_1 : \theta = \theta_1,
\eel
\end{linenomath*}
where $H_0$ and $H_1$ are the \textit{null} and \textit{alternative hypotheses}, respectively. \textit{Simple} here refers to the fact that both $H_0$ and $H_1$ correspond to single $\theta$ values which completely determine the distributions $\mathbb{P}_{\theta_0}$ and $\mathbb{P}_{\theta_1}$. The aim is to decide whether to reject $H_0$ in favor of $H_1$, which is done by choosing a subset $\mathcal{R}$ of $\mathcal{Y}$ called the \textit{rejection region}. This choice of $\mathcal{R}$ completely determines the test: If $Y\in\mathcal{R}$, then reject $H_0$ in favor of $H_1$; if $Y\notin\mathcal{R}$, then do not reject $H_0$. \textit{Type I error} occurs when $H_0$ is true but is rejected, and \textit{type II error} occurs when $H_0$ is false but not rejected; see Table \ref{tab:hypothesis}. This is quantified\footnote{Type I and II error rates are commonly denoted by $\alpha$ and $\beta$, but since $\beta$ is already used as an SIR parameter we adopt the unconventional notation $\mathcal{E}_1$ and $\mathcal{E}_2$.} as
\begin{linenomath*}
\be
\begin{aligned}
	&\mathcal{E}_1(\mathcal{R}) = \text{\textit{Type I error rate}}
		= \mathbb{P}_{\theta_0}(Y\in\mathcal{R})
		= \mathbb{P}_{\theta_0}(\text{Reject}\ H_0), \\
		&\mathcal{E}_2(\mathcal{R}) = \text{\textit{Type II error rate}} 
		= \mathbb{P}_{\theta_1}(Y\notin\mathcal{R})
		= \mathbb{P}_{\theta_1}(\text{Do not reject}\ H_0). 
\end{aligned}
\ee
\end{linenomath*}
Ideally one would find a rejection region $\mathcal{R}$ that simultaneously minimizes type I and type II error rates, but this is generally impossible. Instead, a common statistical paradigm is to fix a \textit{significance level} $\alpha>0$ and minimize $\mathcal{E}_2(\mathcal{R})$ subject to the constraint $\mathcal{E}_1(\mathcal{R})=\alpha$. For such an $\alpha$, a region $\mathcal{R}$ is called \textit{a most powerful level-$\alpha$ rejection region} if $\mathcal{E}_1(\mathcal{R})=\alpha$ and $\mathcal{E}_2(\mathcal{R})\leq \mathcal{E}_2(\mathcal{R}')$ for all $\mathcal{R}'$ satisfying $\mathcal{E}_2(\mathcal{R}')=\alpha$. That is, $\mathcal{R}$ minimizes type II error over all rejection regions with type I error equal to $\alpha$. The Neyman-Pearson Lemma gives the most powerful rejection region in the case of a simple hypothesis test.
\begin{lemma}\label{neyman-pearson}
(Neyman-Pearson) Let $L(Y\vert \theta)$ denote the likelihood function for data $Y$ and a parameter $\theta$, and fix $\alpha>0$. Then there exists an $\eta\in\mathbb{R}$ such that
\begin{linenomath*}
\bel\label{eq:LRT}
	\mathcal{R}_{LR} = \bigg\{Y : \frac{L(Y\vert\theta_1)}{L(Y\vert\theta_0)} \geq \eta\bigg\}
\eel
\end{linenomath*}
is a most powerful level-$\alpha$ rejection region for the hypothesis test \eqref{eq:hypothesis}.
\end{lemma}
$L(Y\vert\theta_1)/L(Y\vert\theta_0)$ is called the \textit{likelihood ratio} and the decision to reject or not reject $H_0$ based on the rejection region $\mathcal{R}_{LR}$ is called the \textit{likelihood ratio test}. Since the Neyman-Pearson Lemma guarantees the likelihood ratio test is most powerful in our setting, we henceforth consider only the rejection region $\mathcal{R}_{LR}$ and set $\mathcal{E}_1=\mathcal{E}_1(\mathcal{R}_{LR})$ and $\mathcal{E}_2=\mathcal{E}_2(\mathcal{R}_{LR})$.

\begin{table}
    \centering
    \begin{tabular}{c|c|p{0.2\textwidth}}
      &  $\boldsymbol{H_0:\theta=\theta_0}$ & $\boldsymbol{H_1:\theta=\theta_\epsilon}$ \\
    \hline
    \textbf{Reject} $\boldsymbol{H_0}$ & Type I error & Success \\
    \hline
    \textbf{Do not reject} $\boldsymbol{H_0}$ & Success & Type II error
    \end{tabular}
    \vspace{\baselineskip}
    \caption{\textit{Simple hypothesis test}.}
    \label{tab:hypothesis}
\end{table}

Returning now to the SIR model, our hypothesis test of interest is
\begin{linenomath*}
\bel\label{eq:hypothesisSIR}
H_0 : \theta = \theta_0 \quad\text{ vs. }\quad H_1 : \theta = \theta_{\epsilon}(\omega)
\eel
\end{linenomath*}
where, as before, $\theta_\epsilon(\omega)$ is a perturbation of $\theta$ of size $\epsilon$ in the direction $\omega$.  The observed data are $Y_{1:T} = (Y_1,\dots, Y_T)$ for any time $T$ before the time of peak infection, with each $Y_t=p\Delta_t+\xi_t$ as in \eqref{eq:NoisySIR}. For the rest of this paper $\mathcal{E}_2(\omega)$ will denote the type II error rate of \eqref{eq:hypothesisSIR} for the likelihood ratio test with angle $\omega$. As such, the likelihood ratio test \eqref{eq:LRT} minimizes $\mathcal{E}_2(\omega)$ thereby providing the most powerful technique for detecting differences of order $\epsilon$ in the SIR model parameters. We also set $\Delta_t^\epsilon(\omega)= \Delta_t(\theta_\epsilon(\omega))$ where, recall, $\Delta_t(\theta)=N(s_{t-1}(\theta)-s_t(\theta))$, and let $\Phi$ denote the standard normal cumulative distribution function. With this notation we now present the main result of this subsection.

\begin{approximation}\label{prop2}
For any $\epsilon>0$, $\omega\in[0,2\pi)$, and significance level $\alpha>0$,
\begin{linenomath*}
\footnotesize
\begin{align}
	\mathcal{E}_2(
	\omega) &\approx  1-\Phi\left(\Phi^{-1}(\alpha)+pNi_0\sqrt{\sum_{t=1}^T\frac{e^{2\delta t}}{\sigma_t^2}\left[\beta_\epsilon\bigg(\frac{e^{-\delta_\epsilon}-1}{-\delta_\epsilon}\bigg)e^{\epsilon tf(\omega)}-\beta\bigg(\frac{e^{-\delta}-1}{-\delta}\bigg)\right]^2}\phantom{-}\right)\label{eq:approx1} \\
	    &\approx 1-\Phi\left(\Phi^{-1}(\alpha)+pNi_0\sqrt{\sum_{t=1}^T\frac{e^{2\delta t}}{\sigma_t^2}\left[(\beta+\epsilon\cos\omega) e^{\epsilon tf(\omega)}-\beta\right]^2}\phantom{-}\right),\label{eq:approx2}
\end{align}
\end{linenomath*}
\normalsize
where $f(\omega)=\cos(\omega)-\sin(\omega)$, $\beta_\epsilon=\beta+\epsilon\cos\omega$, and $\delta_\epsilon=\delta+\epsilon f(\omega)$. Moreover,
\begin{linenomath*}
\be
    \mathcal{E}_2(\pi/4)\approx\mathcal{E}_2(5\pi/4)\approx\sup_{\omega\in[0,2\pi)}\mathcal{E}_2(\omega).
\ee
\end{linenomath*}
\end{approximation}

\begin{figure}[!ht]
\centering
\begin{tabular}{c}
\includegraphics[width=1\textwidth]{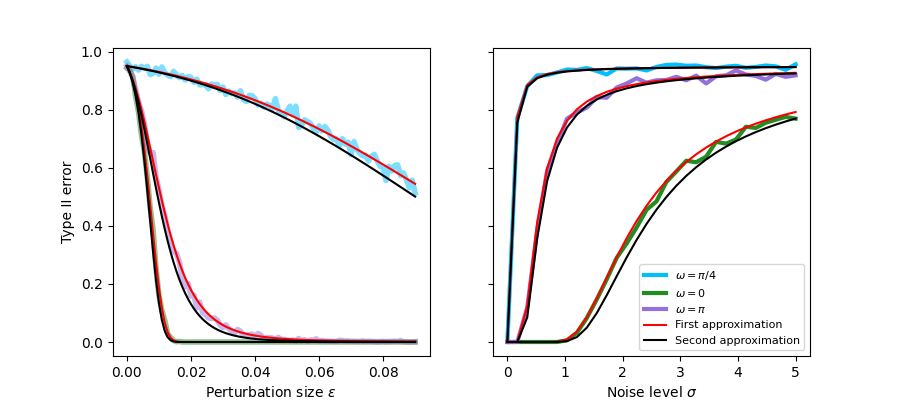}
\end{tabular}
\caption{\textit{Type II error as a function of perturbation size and noise level}. The left panel shows the empirical and theoretical type II errors for the angles $\omega=0,\pi/4$, and $\pi$ as a function of perturbation size $\epsilon$ with fixed noise level $\sigma=0.3$. The right panel shows the empirical and theoretical type II errors for the same angles as a function of noise level $\sigma$ with fixed perturbation size $\epsilon=.03$. In each case the SIR parameters are those from Section \ref{model}, namely $(\beta,\gamma)=(.21,.07)$, $N=10^7$, and initial condition $i_0=1/N$. The time horizon $T$ is $60$ days into the epidemic, which in this case is $60$ days prior to the peak time. The significance level is $\alpha=.05$. Here, \textit{theoretical} refers to the first (red) and second (black) approximations of type II error $\mathcal{E}_2(\omega)$ in \edit{Approximation} \ref{prop2}, i.e. equations \eqref{eq:approx1} and \eqref{eq:approx2}, respectively. \textit{Empirical} refers to the type II error obtained by performing 1000 simulations of the noisy SIR model \eqref{eq:NoisySIR} followed by a likelihood ratio test of the hypothesis in \eqref{eq:hypothesisSIR} for each set of parameters. More specifically, the red and black curves lying over the blue line are the type II error approximations \eqref{eq:approx1} and \eqref{eq:approx2} when $\omega=\pi/4$, those lying over the purple line are when $\omega=\pi$, and those lying over the green line are when $\omega=0$, with the blue, green, and purple curves corresponding to the empirically computed type II error rates when $\omega=\pi/4, 0$, and $\pi$, respectively. In each case both theoretical results closely align with the empirical ones, with the first approximation being slightly better than the second as expected. Also as predicted, the empirical type II errors all approach $1-\alpha=.95$ both as perturbation size goes to $0$ and as the noise level gets large, and this approach is most rapid when $\omega=\pi/4$. In each case the noise model is Case 2, $\sigma_t=N\sigma i_t$. For the simulated blue, green, and purple curves, we used a numerical integrator to obtain the $i_t$ values, while for the red and black curves we used the pre-peak approximation $i_t\approx e^{\delta t}i_0$.}
\label{fig3}
\end{figure}

The derivation of \edit{Approximation} \ref{prop2} is in the Appendix. The first and second approximations of $\mathcal{E}_2(\omega)$ correspond to the red and black curves in Figure \ref{fig3}, respectively. Comparing these to the empirical type II error rates (the blue, green, and purple curves) we see these approximations are sound. In particular, the last part of \edit{Approximation} \ref{prop2} indicates the angles $\pi/4$ and $5\pi/4$ give rise to the largest type II error rate for the hypothesis test \eqref{eq:hypothesisSIR} with perturbation size $\epsilon$ and significance level $\alpha$. To quantify the magnitude of type II error in these cases, we substitute into the second approximation to get
\begin{linenomath*}
\bel\label{eq:pi4}
	\mathcal{E}_2(\tfrac{\pi}{4}) \approx \mathcal{E}_2(\tfrac{5\pi}{4})
		\approx 1 - \Phi\left(\Phi^{-1}(\alpha)+\frac{pNi_0\epsilon}{\sqrt{2}}\sqrt{\sum_{t=1}^T\frac{e^{2\delta t}}{\sigma_t^2}}\phantom{-}\right).
\eel
\end{linenomath*}
Note that as the noise level $\sigma_t^2$ goes to $0$, the sum under the square root goes to infinity and the entire expression goes to $1-\Phi(\infty)=0$. That is, if there is no noise then the type II error rate of the likelihood ratio test will vanish. If there is any noise at all however, the sum is finite and \eqref{eq:pi4} becomes arbitrarily close to $1-\Phi(\Phi^{-1}(\alpha))=1-\alpha$ as either $p$, the probability of detecting an infected individual, or $\epsilon$, the perturbation size, go to $0$. For example, if we set the type I error rate to $\alpha=0.1$ then as either $p$ or $\epsilon$ go to $0$, the probability of making a type II error will approach $0.9$. Similarly, type II error will go to $1-\alpha$ as $\sigma_t^2$ goes to infinity. This limit is unrealistic though since $\sigma_t^2$ is the variance of observed data and as such should be less than the population size. This leads us to consider two cases for noise. 
\vspace*{.2cm}
\begin{enumerate}[label={\textit{Case \arabic*.}}, align=left]
	\item \textit{Noise proportional to population size, i.e. $\sigma_t = N\sigma$ for $\sigma$ in $(0,1)$}.
	\vspace*{.1cm}
	\item \textit{Noise proportional to number of infections, i.e. $\sigma_t= N\sigma i_t$ for $\sigma>0$}.
\end{enumerate}
\vspace*{.2cm}
In both cases $\sigma$ is constant and independent of $t$. Case 2 involves $i_t$ which is not expressible in closed-form. However, we can use \edit{Approximation} \ref{prop1} and its derivation, specifically the approximate solution \eqref{approxsol}, to circumvent this issue by replacing $i_t$ with $e^{\delta t}i_0$. As discussed in Section \ref{bounds}, this approximation is appropriate early in the epidemic. In Case 1, Equation \eqref{eq:pi4} becomes
\begin{linenomath*}
\be
	\mathcal{E}_2 \approx\ 1 - \Phi\left(\Phi^{-1}(\alpha)+\frac{pi_0\epsilon}{\sigma\sqrt{2}}\sqrt{\sum_{t=1}^T e^{2\delta t}}\phantom{-}\right).
\ee
\end{linenomath*}
In addition to the aforementioned limits, we see in this case that the expression, and hence the type II error, approaches $1-\alpha$ as the population $N$ goes to infinity (so that $i_0=1/N$ goes to $0$). In Case 2, Equation \eqref{eq:pi4} becomes
\begin{linenomath*}
\bel\label{eq:10}
	\mathcal{E}_2(\tfrac{\pi}{4}) \approx \mathcal{E}_2(\tfrac{5\pi}{4})  
		\approx 1 - \Phi\left(\Phi^{-1}(\alpha)+\frac{p\epsilon\sqrt{T}}{\sigma \sqrt{2}}\right).
\eel
\end{linenomath*}
The above expression does not depend on population size, $N$, nor on the SIR parameters $\beta$ and $\gamma$, while the asymptotic results for $p$, $\epsilon$, and $\sigma$ still apply. Since Case 1 has noise proportional only to $N$, it implicitly assumes relative noise is larger earlier in the outbreak which may not be realistic. Case 2 avoids this since relative noise will be small whenever the reported number of infected individuals is small, e.g. early in an epidemic. For this reason and its invariance under different model parameters and population sizes, we consider only Case 2 moving forward.


\subsection{Simple illustration: Implications of \edit{Approximation} \ref{prop2}} \label{implications}
\label{sec:Emp_Analysis}

\edit{Approximation} \ref{prop2} says that given an SIR parameter $\theta_0$, the probability of failing to reject the hypothesis $H_0:\theta = \theta_0$ when the alternative $H_1:\theta=\theta_\epsilon$ is true can be very high, especially when the angle of perturbation is $\pi/4$ or $5\pi/4$. In this subsection we take a closer look at what this means for epidemic prediction.

Consider for concreteness the familiar setting $\theta_0=(.21, .07)$, $N=10^7$, and $i_0=1/N$. Figure \ref{fig4} shows the total number of infections $10$ days past the peak time as well as the duration\footnote{We define the duration of the epidemic to be first day after the peak time such that less than $10$ individuals are infected.} of the epidemic for $\theta_0$, $\theta_\epsilon(\pi/4)$, and $\theta_\epsilon(5\pi/4)$ and varying perturbation sizes $\epsilon$. 

\begin{figure}[!ht]
\centering
\begin{tabular}{c}
\includegraphics[width=1\textwidth]{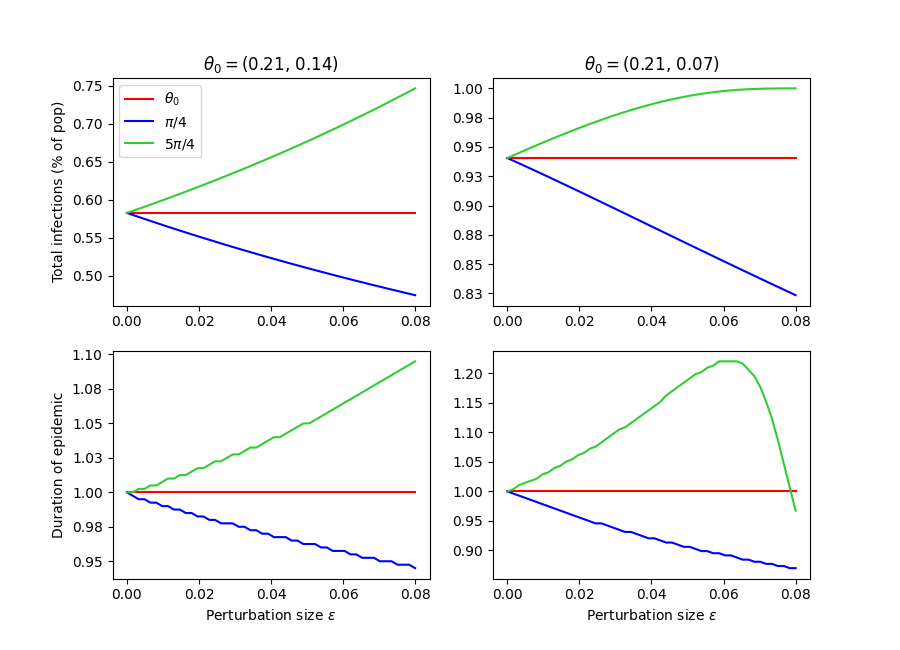}
\end{tabular}
\caption{\textit{Consequences of type II error}. The top panels show the total number of infections $10$ days past the time of peak infection as a percentage of the total population. The bottom panels show the duration of the epidemic, which is defined to be the first day past the peak when less than $10$ individuals are infectious. The left panels correspond to the parameter $\theta_0=(.21,.14)$ and the right panels to $\theta_0=(.21,.07)$. The red lines give the total percent infected or duration of the epidemic for the true parameter $\theta_0$ in each of their respective plots, while the blue and green curves give these values for $\theta_\epsilon(\pi/4)$ and $\theta_\epsilon(5\pi/4)$ over a range of $\epsilon$ values, respectively. In all cases $N=10^7$.}
\label{fig4}
\end{figure}

Setting $\omega=\pi/4$ or $5\pi/4$ and letting noise be as in Case 2, the first approximation in \edit{Approximation} \ref{prop2} can be rearranged to obtain
\begin{linenomath*}
\be
	\epsilon  \approx \frac{\big[\Phi^{-1}(1-\mathcal{E}_2)-\Phi^{-1}(\alpha)\big]\sigma\delta e^\delta\sqrt{2}}{(e^\delta-1)p\sqrt{T}}.
\ee
\end{linenomath*}
From this we can compute the consequences of type II error. For example, suppose we are $60$ days into an epidemic $(T=60)$ and wish to test the hypothesis $\theta_0=(.21,.07)$ versus $\theta_\epsilon(5\pi/4)$ as above. Moreover, suppose $p=1$ (perfect diagnostics), $\sigma = 0.2$ (infection standard deviation of $\pm 20\%$ of new cases), and $\alpha=.05$ and $\mathcal{E}_2=0.5$. Then the above equation gives $\epsilon \approx .064$. Thus, reading off the right panels in Figure \ref{fig4}, we see that under these fairly generous conditions a type II error -- which has a $50\%$ chance of occurring -- will result in underestimating the total number of infections of an epidemic by over $5\%$ of the total population and the duration of an epidemic by over $20\%$ of the predicted duration. For $\pi/4$ a type II error in this setting will result in overestimating the total infections by nearly $10\%$ of the total population and the duration by approximately $10\%$ of the predicted one. Similarly, $\epsilon\approx .062$ when $\theta_0=(.21,.14)$ with all other parameters the same, and again from the left panels in Figure \ref{fig4} we observe significantly different predicted outcomes depending on whether or not the null hypothesis $H_0:\theta=\theta_0$ is rejected.

\begin{figure}[!ht]
\centering
\includegraphics[width=1\textwidth]{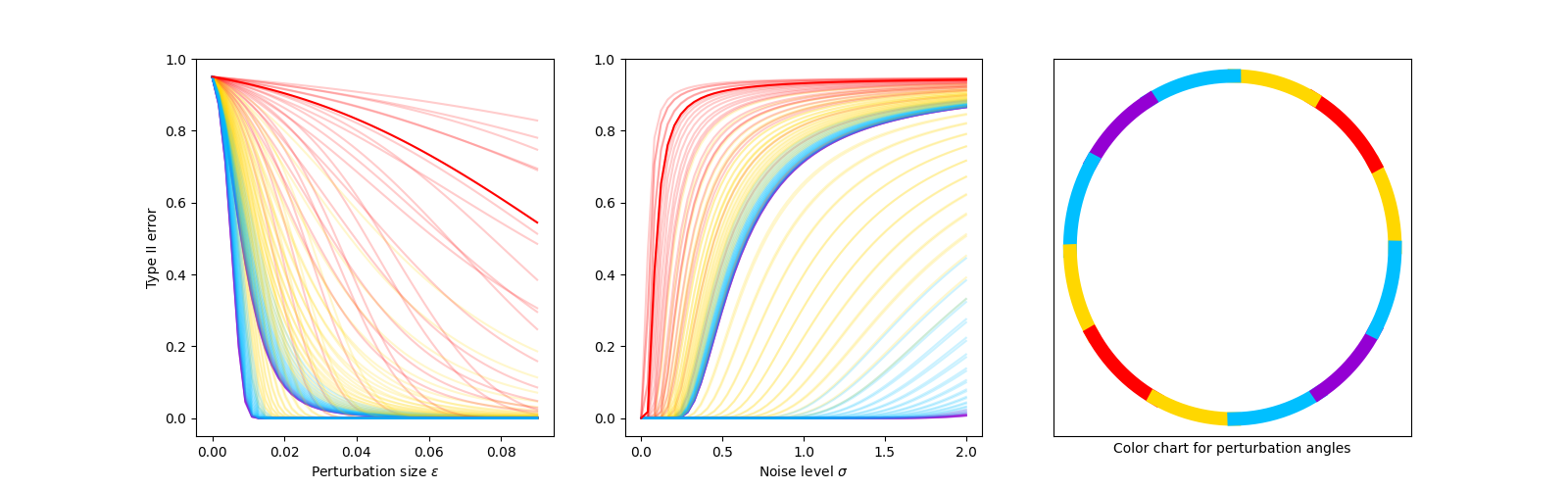}
\caption{\textit{Practical identifiability of $\delta$}. The left and center panels use the first type II error approximation in \edit{Approximation} \ref{prop2} to graph type II error as a function of perturbation size, $\epsilon$, and noise level, $\sigma$, respectively. Each of the rainbow colored curves in both panels correspond to one of $150$ different values of $\omega$ spread uniformly across $[0,2\pi)$. The color chart in the right panel indicates the colors corresponding to different angles $\omega$: The light red curves correspond to the $\omega$ closest to $\pi/4$ and $5\pi/4$, the yellow are those a bit farther away, the blue still farther, and the purple are those farthest from $\pi/4$ and $5\pi/4$, i.e. closest to $3\pi/4$ and $7\pi/4$. Finally, the dark red curve in each of the two panels corresponds to $\omega=\pi/4$ and $5\pi/4$, which have the same type II error. Note the rapid fall off in type II error as angles get farther from $\pi/4$ and $5\pi/4$, especially as a function of $\epsilon$. This agrees with the empirical observation in Figure \ref{fig:SIRExample} that MLE favors parameters lying along a line of slope $1$. In particular, the inverse problem for $\delta$ is practically identifiable.}
\label{fig5}
\end{figure}

While our empirical and theoretical results indicate the inverse problem of finding $\theta=(\beta,\gamma)$ is prone to error, they also show inference of $\delta$ is robust and reliable (see for instance Figures \ref{fig:SIRExample} and \ref{fig5}). In particular, since $\theta=(\beta,\gamma)$ is completely determined by\footnote{We choose to focus on $\gamma$ because, unlike $\beta$ which depends on the average number of people an infectious person will come in contact with, $\gamma$ depends only on the pathogen, not on human social behavior, and therefore tends to be more stable and better approximated in practice.} $\gamma$ and $\delta$, knowledge of $\delta$ reduces the inverse problem to finding $\gamma$, the reciprocal of the average number of days an individual is infectious. In this case our new hypothesis test becomes
\begin{linenomath*}
\bel\label{eq:hypothesisGamma}
	H_0: \gamma = \gamma_0 \quad\text{ vs. }\quad H_1 : \gamma = \gamma_0+\hat\epsilon.
\eel
\end{linenomath*}
for some real number $\hat\epsilon$. Furthermore, knowing $\delta$ implies $\theta$ lies on the line of slope $1$ with vertical intercept $-\delta$. So by a simple geometric argument (see Figure \ref{fig6}), the above hypothesis test is equivalent to the hypothesis test \eqref{eq:hypothesisSIR} with $\epsilon=\lvert\hat\epsilon\rvert\sqrt{2}$ and angle $\pi/4$ if $\hat\epsilon>0$ or $5\pi/4$ if $\hat\epsilon<0$. So by \eqref{eq:10} the type II error of \eqref{eq:hypothesisGamma} is
\begin{linenomath*}
\be
	\mathcal{E}_2 = 1 - \Phi\left(\Phi^{-1}(\alpha)+\frac{p\epsilon\sqrt{T}}{\sigma \sqrt{2}}\right)
		= 1 - \Phi\left(\Phi^{-1}(\alpha)+\frac{p\lvert\hat\epsilon\rvert\sqrt{T}}{\sigma}\right).
\ee
\end{linenomath*}
Thus rather than consider the original hypothesis test, one can first infer $\delta$, then consider the hypothesis test \eqref{eq:hypothesisGamma} with type II error rate as above.

\begin{figure}[!ht]
\centering
\begin{tabular}{c}
\includegraphics[width=.5\textwidth, height=.25\textheight]{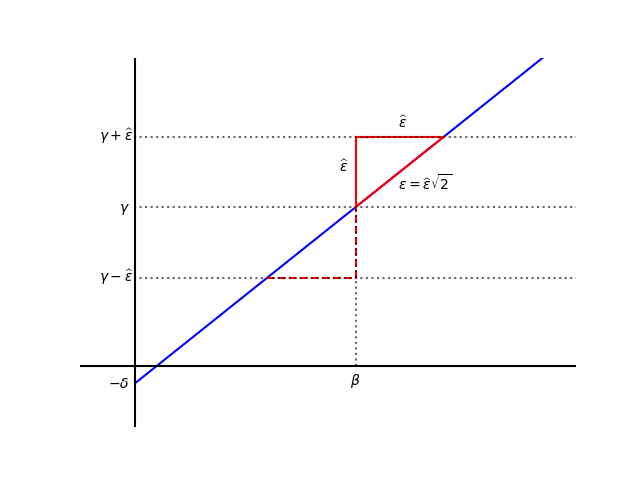}
\end{tabular}
\caption{\textit{Going from hypothesis test \eqref{eq:hypothesisGamma}} to \eqref{eq:hypothesisSIR}.}
\label{fig6}
\end{figure}


\subsection{Empirical analysis: NYC Covid cases, March 2020}\label{NYC}


In this section we discuss the extension of the theoretical results on parametric non-identifiability to a real world dataset. Consider the Spring 2020 COVID-19 outbreak in New York City. The New York City Health Department keeps a repository of all public COVID-19 data online \cite{NYCHD2021}. Using their daily case data as a proxy for new infections, we directly apply equation (\ref{eq:NoisySIR}) to the noisy data. To focus on estimation early in the pandemic, we focus on reported daily cases from February 29, 2020 through March 14, 2020, which approximately represent the first two weeks of the pandemic in New York City. This period precedes the statewide lockdown including the closing of schools on March 15th. However, the increasing awareness of COVID-19 and increased testing capacity strongly suggest that the contact rate $\beta$ and reporting rate $p$ were likely non-constant during this time. These parameters are also not jointly identifiable. Thus, we make the simplifying assumption that they are constant.  Below, we show estimates of $\beta$, $\gamma$, and $\sigma$ for fixed values of $p$ ranging from 0.01 to 0.25 consistent with the current literature on the underreporting rates of COVID-19 infection \cite{DEOLIVEIRA,Richterich2020,LAU2021110}.

To connect with the earlier analysis, we are following Case 2 as discussed in section 3.2, in which the noise is proportional to number of infections, i.e. $\sigma_t=N\sigma i_t$ for some $\sigma > 0$ which is also inferred via maximum likelihood. We thus model daily infections by
\begin{linenomath*}
\be
y_t = pN(s_{t-1}-s_t) + \sqrt{N i_t}\epsilon_t, \quad \epsilon_t \sim N(0,\sigma^2)
\ee
\end{linenomath*}
from which we obtain the log-likelihood function
\begin{linenomath*}
\be
    \ell(y_t\vert\beta,\gamma,\sigma) \approx -\frac{1}{2}\sum_{k=1}^{t} \frac{\left(y_k-pN(s_{k-1}-s_{k})\right)^2}{Ni_t\sigma^2}.
\ee
\end{linenomath*}
For a fixed value of $p=0.05$, maximizing the above likelihood gives estimates
\begin{linenomath*}
\be
(\hat{\beta},\hat{\gamma},\hat{\sigma}) = (4.82,4.22,1.37)
\ee
\end{linenomath*}
and a corresponding estimate of $\hat{R}_0 = 1.14$. The  SIR curve generated by the maximum likelihood estimates of $\beta$ and $\gamma$ is shown in Figure \ref{fig:NYCdata_with_MLEGenSIR}  with corresponding 95\% confidence regions based on the maximum likelihood estimate of $\sigma.$ Additional results for $p=0.01, 0.02$, and $0.1$ are also shown in the figure.

\begin{figure}[!ht]
    \centering
    \includegraphics[width=1.0\textwidth]{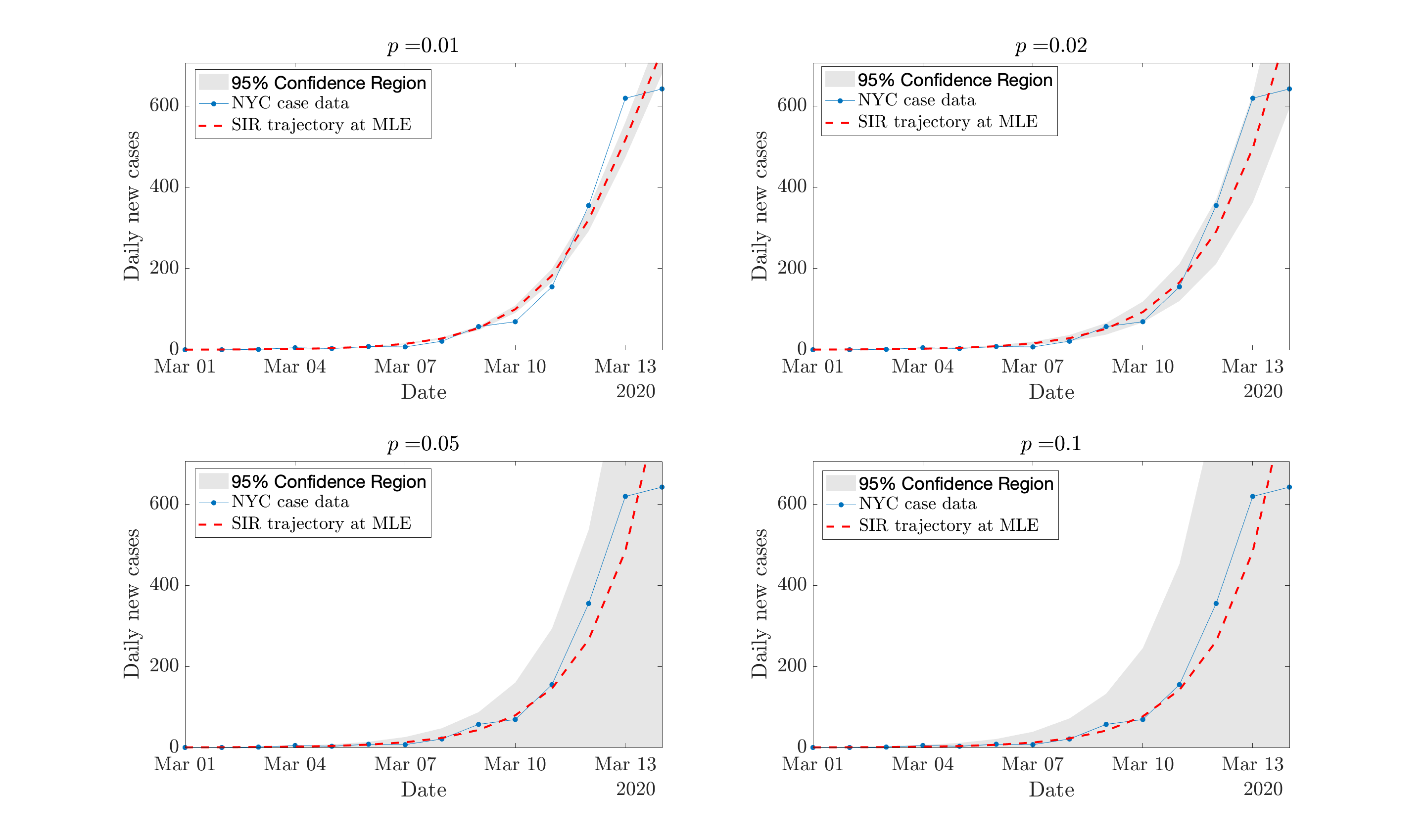}
    \caption{New York City public testing results for COVID-19 from the first known case on February 29, 2020 to March 15, 2020. We have used maximum likelihood estimation to generate an SIR trajectory through the noisy data for each reporting rate.}
    \label{fig:NYCdata_with_MLEGenSIR}
\end{figure}

Returning to the testing framework, Figure \ref{fig:Case2_power} provides type II error estimates based on the approximation of Equation (\ref{eq:10}) with significance level $\alpha = 0.1$, reporting rates $p=0.01$, $0.02$, $0.05$, and $0.1$, and $T = 14$ days of new infection counts. For all values of $p$ considered herein, the MLE of $\hat{\sigma}$ is greater than 0.75 and corresponds with Type II error greater than 80\% for all values of $\epsilon$ such that $\hat{\theta}_\epsilon(\pi/4)$ or $\hat{\theta}_\epsilon(5\pi/4)$ with corresponding $R_0 > 1$.

\begin{figure}[!ht]
    \centering
    \includegraphics[width=1.0\textwidth]{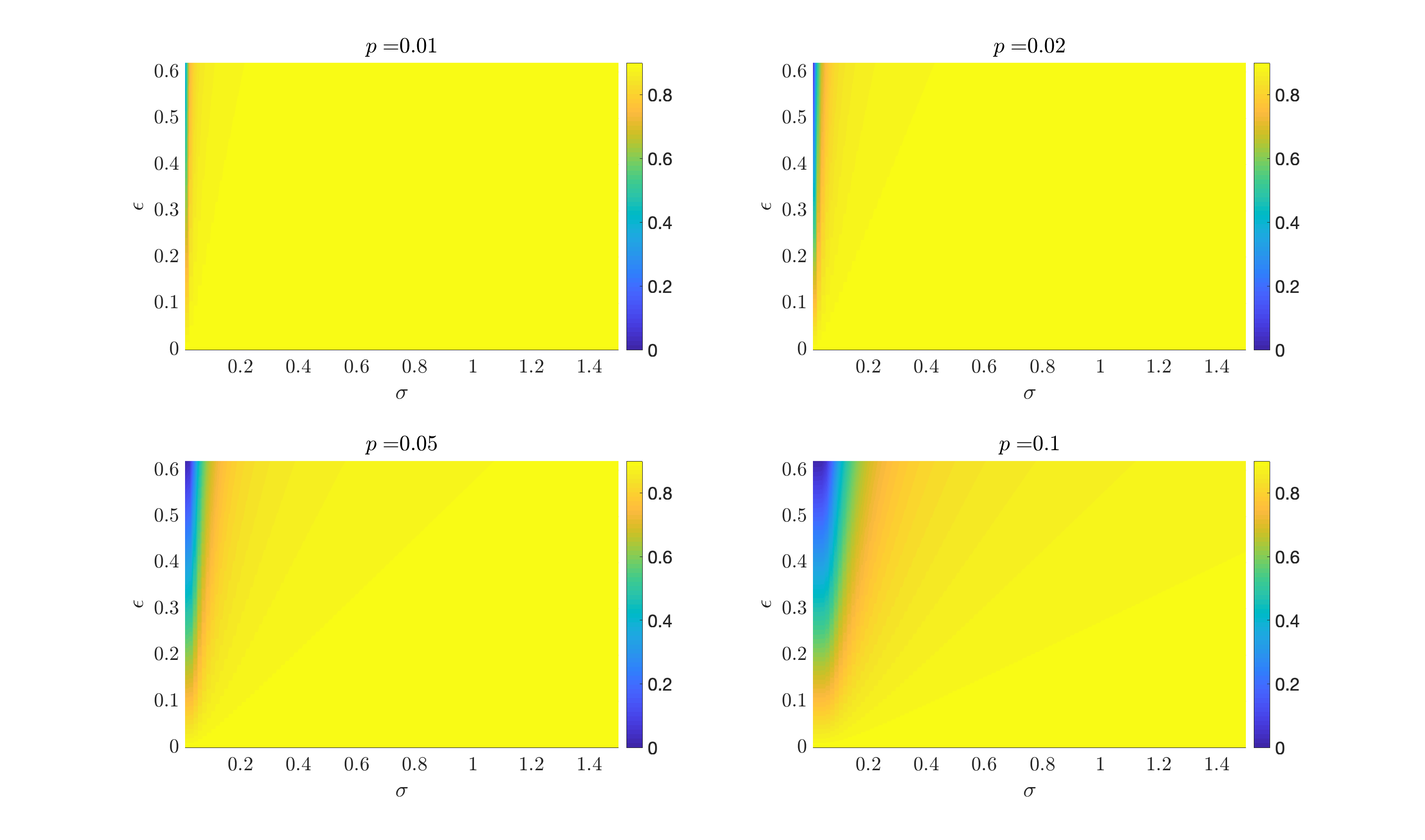}
    \caption{Type II error rate as a function of $\epsilon$ and $\sigma$ at significance level $\alpha =0.1$, reporting rate $p=0.15$, and $T=14$ days of new infection observations.}
    \label{fig:Case2_power}
\end{figure}

Thus, while there may be a large disparity between the true SIR parameters and our maximum likelihood estimates -- hence large differences in the estimate of $R_0$ -- the hypothesis testing framework has very low power to detect such differences. This result is based on the most difficult to detect perturbations in $\theta$. However, it provides pessimistic but important lower bounds on the extent to which one can rely on parameter estimates from noisy, early pandemic data.

One final note on the preceding example. The MLEs of $\beta$ and $\gamma$ in the previous analysis are quite sensitive to the reporting rate.  For reference, Table \ref{tab:NYC_params_by_p} provides corresponding MLE estimates for $\beta$, $\gamma$, and $\sigma$ as a function of $p$.
\begin{table}
\centering
\begin{tabular}{c|c| c | c | c }
     $p$ & $\hat{\beta}$ & $\hat{\gamma}$ & $\hat{\sigma}$ & $\hat{R}_0$    \\
     \hline
     0.01 & 19.69& 19.03 & 0.83 & 1.03 \\
     0.02 & 15.59 & 14.99 & 1.29 & 1.04\\
     0.03 & 9.54 & 8.95 & 1.38 & 1.07 \\
     0.04 & 6.43 & 5.82 & 1.38 & 1.10 \\
     0.05 & 4.82 & 4.22 & 1.37 & 1.14\\
     0.1 & 2.20 & 1.58 & 1.36 & 1.39\\
     0.15 & 1.44 & 0.82 & 1.35 & 1.7\\
     0.2 & 1.07 & 0.46 & 1.35 & 2.35 \\
     0.25 & 0.86 & 0.24 & 1.36 & 3.58
\end{tabular}
\vspace{\baselineskip}
\caption{Maximum likelihood estimates of $\beta$, $\gamma$, and $\sigma$ are shown for different choices of reporting rate $p$.  The corresponding estimates of $R_0$ using the MLEs is also provided.}
\label{tab:NYC_params_by_p}
\end{table}
However, the type II error plot in Figure \ref{fig:Case2_power} is largely unchanged for the range of $p$ in the preceding table.  Since $\sigma$ is fairly robust to different choices of $p$, our conclusion about the limited power of testing holds true for the range of $p$ considered.  Therefore,  one has limited statistical power to detect large differences in SIR model parameters in the worst case scenario, regardless of the choice of reporting rate. As such, we believe this article serves as a cautionary tale to those fitting SIR-type models in the early days of an epidemic.

There is an important distinction to make. Having low power to detect a difference is not equivalent to being unable to tell that there is a difference.  Certain parameter values are essentially impossible given natural assumptions about the dynamics of a pandemic. For example, $1/\gamma$ is the average time an infected individual can spread the disease before they are either recovered or removed from the population by quarantine.  Extremely large values of $\gamma$ and hence small values of $1/\gamma$, such as those attained in Table \ref{tab:NYC_params_by_p}, are likely unrealistic.   Thus, the inclusion of side or prior information on $\gamma$ and/or $\beta$ akin to the analysis in Section \ref{sec:Emp_Analysis} can greatly improve one's ability to disambiguate different SIR parameters.


\section{Discussion}\label{discussion}

The preceding analysis was based on a simple implementation of the SIR model.  Practitioners studying future outbreaks may consider a multitude of modifications to our model construction which result in different likelihood functions. Thus, we have decided to conclude this article with a short discussion of how one may adapt our techniques to these different settings to better understand issues of practical identifiability with \textit{noisy} or sparse observations.  

To construct an analytically tractable approximation to the type II error, we assumed the proportion of susceptible individuals remains essentially 1 and thus obtained a linear system, namely \eqref{approxODE}, that approximates the SIR equations. Such approximations are suitable locally in time and are therefore appropriate when one is focused on the early stages of an outbreak. Importantly, a similar approach can be used to construct analytic approximations to the dynamics of any epidemic model. Such expressions will depend on unknown SIR parameters, fixed parameters such as population size, and other parameters such as reporting rate or behavioral factors, as in \cite{Cori}. For example, Britton and Scalia Tomba \cite{Britton} assume the proportion of susceptible individuals remains 1 early in an epidemic to study the problem of inferring infection rate from observations of generation and serial times, which are often available via contact tracing. In all cases one can investigate the use of this and other realistic simplifications of the dynamics to approximate type II error and better understand potential limitations of their particular model. We believe this approach remains an interesting, potentially fruitful avenue toward understanding identifiability in a wide array of epidemic models.

\edit{On the theoretical side, the upper bound for the error term in Proposition \ref{prop:theory} of Appendix \ref{sec:theory} gives some rigorous justification for the approximate dynamics used throughout this work. However, the numerical results of the main text and Appendix \ref{sec:error} indicate the approximation is more accurate than the theoretical bound suggests. It is therefore an open question whether our theoretical bound on the approximation error can be improved upon, perhaps via other approximate solutions of the SIR model found in, for example, \cite{Turkyilmazoglu, barlow, schlickeiser} and references therein. Finally, since the approximation of $s$ by 1 is used in other models and to investigate other questions about epidemics \cite{Sauer2020, Britton}, it is also of interest whether estimates of the error in our setting can be used to control error for similar approximations in related settings.}


\section*{Statements and declarations}

\noindent The authors declare no competing interests.


\section*{Acknowledgments}

\noindent This work was partially funded by NIH R01-ES028804 from the National Institute of Environmental Health Sciences, and the Duke University DOMath summer program. OM also thanks
NSF-DMS-2038056 for partial support during this project.


\bibliographystyle{plain}
\bibliography{ref}


\begin{thebibliography}{40}
\ifx \bisbn   \undefined \def \bisbn  #1{ISBN #1}\fi
\ifx \binits  \undefined \def \binits#1{#1}\fi
\ifx \bauthor  \undefined \def \bauthor#1{#1}\fi
\ifx \batitle  \undefined \def \batitle#1{#1}\fi
\ifx \bjtitle  \undefined \def \bjtitle#1{#1}\fi
\ifx \bvolume  \undefined \def \bvolume#1{\textbf{#1}}\fi
\ifx \byear  \undefined \def \byear#1{#1}\fi
\ifx \bissue  \undefined \def \bissue#1{#1}\fi
\ifx \bfpage  \undefined \def \bfpage#1{#1}\fi
\ifx \blpage  \undefined \def \blpage #1{#1}\fi
\ifx \burl  \undefined \def \burl#1{\textsf{#1}}\fi
\ifx \doiurl  \undefined \def \doiurl#1{\url{https://doi.org/#1}}\fi
\ifx \betal  \undefined \def \betal{\textit{et al.}}\fi
\ifx \binstitute  \undefined \def \binstitute#1{#1}\fi
\ifx \binstitutionaled  \undefined \def \binstitutionaled#1{#1}\fi
\ifx \bctitle  \undefined \def \bctitle#1{#1}\fi
\ifx \beditor  \undefined \def \beditor#1{#1}\fi
\ifx \bpublisher  \undefined \def \bpublisher#1{#1}\fi
\ifx \bbtitle  \undefined \def \bbtitle#1{#1}\fi
\ifx \bedition  \undefined \def \bedition#1{#1}\fi
\ifx \bseriesno  \undefined \def \bseriesno#1{#1}\fi
\ifx \blocation  \undefined \def \blocation#1{#1}\fi
\ifx \bsertitle  \undefined \def \bsertitle#1{#1}\fi
\ifx \bsnm \undefined \def \bsnm#1{#1}\fi
\ifx \bsuffix \undefined \def \bsuffix#1{#1}\fi
\ifx \bparticle \undefined \def \bparticle#1{#1}\fi
\ifx \barticle \undefined \def \barticle#1{#1}\fi
\bibcommenthead
\ifx \bconfdate \undefined \def \bconfdate #1{#1}\fi
\ifx \botherref \undefined \def \botherref #1{#1}\fi
\ifx \url \undefined \def \url#1{\textsf{#1}}\fi
\ifx \bchapter \undefined \def \bchapter#1{#1}\fi
\ifx \bbook \undefined \def \bbook#1{#1}\fi
\ifx \bcomment \undefined \def \bcomment#1{#1}\fi
\ifx \oauthor \undefined \def \oauthor#1{#1}\fi
\ifx \citeauthoryear \undefined \def \citeauthoryear#1{#1}\fi
\ifx \endbibitem  \undefined \def \endbibitem {}\fi
\ifx \bconflocation  \undefined \def \bconflocation#1{#1}\fi
\ifx \arxivurl  \undefined \def \arxivurl#1{\textsf{#1}}\fi
\csname PreBibitemsHook\endcsname

\bibitem{Kermack1927}
\begin{barticle}
\bauthor{\bsnm{Kermack}, \binits{W.O.}},
\bauthor{\bsnm{Mckendrick}, \binits{A.G.}}:
\batitle{{A contribution to the mathematical theory of epidemics}}.
\bjtitle{Proc. R. Soc. London. Ser. A, Contain. Pap. a Math. Phys. Character}
\bvolume{115}(\bissue{772}),
\bfpage{700}--\blpage{721}
(\byear{1927}).
\doiurl{10.1098/rspa.1927.0118}
\end{barticle}
\endbibitem

\bibitem{Ross1916}
\begin{barticle}
\bauthor{\bsnm{Ross}, \binits{L.-C.S.R.}}:
\batitle{{An application of the theory of probabilities to the study of a
  priori pathometry.—Part I}}.
\bjtitle{Proc. R. Soc. London. Ser. A, Contain. Pap. a Math. Phys. Character}
\bvolume{92}(\bissue{638}),
\bfpage{204}--\blpage{230}
(\byear{1916}).
\doiurl{10.1098/rspa.1916.0007}
\end{barticle}
\endbibitem

\bibitem{Ross1917}
\begin{barticle}
\bauthor{\bsnm{Ross}, \binits{L.-C.S.R.}},
\bauthor{\bsnm{Hudson}, \binits{H.P.}}:
\batitle{{An application of the theory of probabilities to the study of a
  priori pathometry.—Part III}}.
\bjtitle{Proc. R. Soc. London. Ser. A, Contain. Pap. a Math. Phys. Character}
\bvolume{93}(\bissue{650}),
\bfpage{225}--\blpage{240}
(\byear{1917}).
\doiurl{10.1098/rspa.1917.0015}
\end{barticle}
\endbibitem

\bibitem{Ross1917a}
\begin{barticle}
\bauthor{\bsnm{Ross}, \binits{R.}},
\bauthor{\bsnm{Hudson}, \binits{H.P.}}:
\batitle{{An application of the theory of probabilities to the study of a
  priori pathometry.—Part II}}.
\bjtitle{Proc. R. Soc. London. Ser. A, Contain. Pap. a Math. Phys. Character}
\bvolume{93}(\bissue{650}),
\bfpage{212}--\blpage{225}
(\byear{1917}).
\doiurl{10.1098/rspa.1917.0014}
\end{barticle}
\endbibitem

\bibitem{Brauer2019}
\begin{bbook}
\bauthor{\bsnm{Brauer}, \binits{F.}},
\bauthor{\bsnm{Castillo-Chavez}, \binits{C.}},
\bauthor{\bsnm{Feng}, \binits{Z.}}:
\bbtitle{{Mathematical Models in Epidemiology}}.
\bsertitle{Texts in Applied Mathematics},
vol. \bseriesno{69}.
\bpublisher{Springer},
\blocation{New York, NY}
(\byear{2019}).
\doiurl{10.1007/978-1-4939-9828-9}.
\burl{http://link.springer.com/10.1007/978-1-4939-9828-9}
\end{bbook}
\endbibitem

\bibitem{Coburn2009}
\begin{barticle}
\bauthor{\bsnm{Coburn}, \binits{B.J.}},
\bauthor{\bsnm{Wagner}, \binits{B.G.}},
\bauthor{\bsnm{Blower}, \binits{S.}}:
\batitle{{Modeling influenza epidemics and pandemics: Insights into the future
  of swine flu (H1N1)}}.
\bjtitle{BMC Med.}
\bvolume{7},
\bfpage{30}
(\byear{2009}).
\doiurl{10.1186/1741-7015-7-30}
\end{barticle}
\endbibitem

\bibitem{Eisenberg2013}
\begin{barticle}
\bauthor{\bsnm{Eisenberg}, \binits{M.C.}},
\bauthor{\bsnm{Robertson}, \binits{S.L.}},
\bauthor{\bsnm{Tien}, \binits{J.H.}}:
\batitle{{Identifiability and estimation of multiple transmission pathways in
  cholera and waterborne disease}}.
\bjtitle{J. Theor. Biol.}
\bvolume{324},
\bfpage{84}--\blpage{102}
(\byear{2013}).
\doiurl{10.1016/j.jtbi.2012.12.021}
\end{barticle}
\endbibitem

\bibitem{Khaleque2017}
\begin{botherref}
\oauthor{\bsnm{Khaleque}, \binits{A.}},
\oauthor{\bsnm{Sen}, \binits{P.}}:
{An empirical analysis of the Ebola outbreak in West Africa}.
Sci. Rep.
\textbf{7}
(2017).
\doiurl{10.1038/srep42594}
\end{botherref}
\endbibitem

\bibitem{Lee2020}
\begin{barticle}
\bauthor{\bsnm{Lee}, \binits{C.}},
\bauthor{\bsnm{Li}, \binits{Y.}},
\bauthor{\bsnm{Kim}, \binits{J.}}:
\batitle{{The susceptible-unidentified infected-confirmed (SUC) epidemic model
  for estimating unidentified infected population for COVID-19}}.
\bjtitle{Chaos, Solitons and Fractals}
\bvolume{139},
\bfpage{110090}
(\byear{2020}).
\doiurl{10.1016/j.chaos.2020.110090}
\end{barticle}
\endbibitem

\bibitem{Pasquali2021}
\begin{botherref}
\oauthor{\bsnm{Pasquali}, \binits{S.}},
\oauthor{\bsnm{Pievatolo}, \binits{A.}},
\oauthor{\bsnm{Bodini}, \binits{A.}},
\oauthor{\bsnm{Ruggeri}, \binits{F.}}:
{A stochastic SIR model for the analysis of the COVID-19 Italian epidemic}
(2021)
{\href{https://arxiv.org/abs/2102.07566}{{arXiv:2102.07566}}}
\end{botherref}
\endbibitem

\bibitem{Rachah2015}
\begin{botherref}
\oauthor{\bsnm{Rachah}, \binits{A.}},
\oauthor{\bsnm{Torres}, \binits{D.F.M.}}:
{Mathematical modelling, simulation, and optimal control of the 2014 ebola
  outbreak in West Africa}.
Discret. Dyn. Nat. Soc.
\textbf{2015}
(2015)
{\href{https://arxiv.org/abs/1503.07396}{{arXiv:1503.07396}}}.
\doiurl{10.1155/2015/842792}
\end{botherref}
\endbibitem

\bibitem{Yang2020}
\begin{botherref}
\oauthor{\bsnm{Yang}, \binits{X.}},
\oauthor{\bsnm{Wang}, \binits{S.}},
\oauthor{\bsnm{Xing}, \binits{Y.}},
\oauthor{\bsnm{Li}, \binits{L.}},
\oauthor{\bsnm{{Da Xu}}, \binits{R.Y.}},
\oauthor{\bsnm{Friston}, \binits{K.J.}},
\oauthor{\bsnm{Guo}, \binits{Y.}}:
{Revealing the Transmission Dynamics of COVID-19: A Bayesian Framework for
  {\$}R{\_}t{\$} Estimation}
(2020)
{\href{https://arxiv.org/abs/2101.01532}{{arXiv:2101.01532}}}
\end{botherref}
\endbibitem

\bibitem{Sauer2020}
\begin{botherref}
\oauthor{\bsnm{Sauer}, \binits{T.}},
\oauthor{\bsnm{Berry}, \binits{T.}},
\oauthor{\bsnm{Ebeigbe}, \binits{D.}},
\oauthor{\bsnm{Norton}, \binits{M.}},
\oauthor{\bsnm{Whalen}, \binits{A.}},
\oauthor{\bsnm{Schiff}, \binits{S.}}:
{Identifiability of infection model parameters early in an epidemic}.
medRxiv,
2020--061520132217
(2020).
\doiurl{10.1101/2020.06.15.20132217}
\end{botherref}
\endbibitem

\bibitem{Hamelin2020}
\begin{botherref}
\oauthor{\bsnm{Hamelin}, \binits{F.}},
\oauthor{\bsnm{Iggidr}, \binits{A.}},
\oauthor{\bsnm{Rapaport}, \binits{A.}},
\oauthor{\bsnm{Sallet}, \binits{G.}},
\oauthor{\bsnm{{Sallet Observability}}, \binits{G.}},
\oauthor{\bsnm{Hamelin}, \binits{F.}},
\oauthor{\bsnm{Iggidr}, \binits{A.}},
\oauthor{\bsnm{Rapaport}, \binits{A.}},
\oauthor{\bsnm{Sallet}, \binits{G.}}:
{Observability, Identifiability and Epidemiology A survey}.
Technical report
(2020).
\url{https://hal.archives-ouvertes.fr/hal-02995562}
\end{botherref}
\endbibitem

\bibitem{BELLMAN1970}
\begin{barticle}
\bauthor{\bsnm{Bellman}, \binits{R.}},
\bauthor{\bsnm{Åström}, \binits{K.J.}}:
\batitle{On structural identifiability}.
\bjtitle{Mathematical Biosciences}
\bvolume{7}(\bissue{3}),
\bfpage{329}--\blpage{339}
(\byear{1970}).
\doiurl{10.1016/0025-5564(70)90132-X}
\end{barticle}
\endbibitem

\bibitem{Chis2011}
\begin{barticle}
\bauthor{\bsnm{Chis}, \binits{O.-T.}},
\bauthor{\bsnm{Banga}, \binits{J.R.}},
\bauthor{\bsnm{Balsa-Canto}, \binits{E.}}:
\batitle{{Structural Identifiability of Systems Biology Models: A Critical
  Comparison of Methods}}.
\bjtitle{PLoS One}
\bvolume{6}(\bissue{11}),
\bfpage{27755}
(\byear{2011}).
\doiurl{10.1371/journal.pone.0027755}
\end{barticle}
\endbibitem

\bibitem{Bellu2007}
\begin{barticle}
\bauthor{\bsnm{Bellu}, \binits{G.}},
\bauthor{\bsnm{Saccomani}, \binits{M.P.}},
\bauthor{\bsnm{Audoly}, \binits{S.}},
\bauthor{\bsnm{D'Angi{\`{o}}}, \binits{L.}}:
\batitle{{DAISY: A new software tool to test global identifiability of
  biological and physiological systems}}.
\bjtitle{Comput. Methods Programs Biomed.}
\bvolume{88}(\bissue{1}),
\bfpage{52}--\blpage{61}
(\byear{2007}).
\doiurl{10.1016/j.cmpb.2007.07.002}
\end{barticle}
\endbibitem

\bibitem{Brunel2008}
\begin{barticle}
\bauthor{\bsnm{Brunel}, \binits{N.J.B.}}:
\batitle{{Parameter estimation of ODE's via nonparametric estimators}}.
\bjtitle{Electron. J. Stat.}
\bvolume{2}(\bissue{March 2007}),
\bfpage{1242}--\blpage{1267}
(\byear{2008}).
\doiurl{10.1214/07-EJS132}
\end{barticle}
\endbibitem

\bibitem{Chapman2009}
\begin{barticle}
\bauthor{\bsnm{Chapman}, \binits{J.D.}},
\bauthor{\bsnm{Evans}, \binits{N.D.}}:
\batitle{{The structural identifiability of susceptible-infective-recovered
  type epidemic models with incomplete immunity and birth targeted
  vaccination}}.
\bjtitle{Biomed. Signal Process. Control}
\bvolume{4}(\bissue{4}),
\bfpage{278}--\blpage{284}
(\byear{2009}).
\doiurl{10.1016/j.bspc.2009.02.003}
\end{barticle}
\endbibitem

\bibitem{Daly2018}
\begin{botherref}
\oauthor{\bsnm{Daly}, \binits{A.C.}},
\oauthor{\bsnm{Gavaghan}, \binits{D.}},
\oauthor{\bsnm{Cooper}, \binits{J.}},
\oauthor{\bsnm{Tavener}, \binits{S.}}:
{Inference-based assessment of parameter identifiability in nonlinear
  biological models}.
J. R. Soc. Interface
\textbf{15}(144)
(2018).
\doiurl{10.1098/rsif.2018.0318}
\end{botherref}
\endbibitem

\bibitem{Piazzola2020}
\begin{barticle}
\bauthor{\bsnm{Piazzola}, \binits{C.}},
\bauthor{\bsnm{Tamellini}, \binits{L.}},
\bauthor{\bsnm{Tempone}, \binits{R.}}:
\batitle{{A note on tools for prediction under uncertainty and identifiability
  of SIR-like dynamical systems for epidemiology}}.
\bjtitle{Math. Biosci.}
(\byear{2020})
{\href{https://arxiv.org/abs/2008.01400}{{arXiv:2008.01400}}}.
\doiurl{10.1016/j.mbs.2020.108514}
\end{barticle}
\endbibitem

\bibitem{Tuncer2016}
\begin{barticle}
\bauthor{\bsnm{Tuncer}, \binits{N.}},
\bauthor{\bsnm{Gulbudak}, \binits{H.}},
\bauthor{\bsnm{Cannataro}, \binits{V.L.}},
\bauthor{\bsnm{Martcheva}, \binits{M.}}:
\batitle{{Structural and Practical Identifiability Issues of
  Immuno-Epidemiological Vector–Host Models with Application to Rift Valley
  Fever}}.
\bjtitle{Bull. Math. Biol.}
\bvolume{78}(\bissue{9}),
\bfpage{1796}--\blpage{1827}
(\byear{2016}).
\doiurl{10.1007/s11538-016-0200-2}
\end{barticle}
\endbibitem

\bibitem{Tuncer2018}
\begin{barticle}
\bauthor{\bsnm{Tuncer}, \binits{N.}},
\bauthor{\bsnm{Le}, \binits{T.T.}}:
\batitle{{Structural and practical identifiability analysis of outbreak
  models}}.
\bjtitle{Math. Biosci.}
\bvolume{299}(\bissue{February}),
\bfpage{1}--\blpage{18}
(\byear{2018}).
\doiurl{10.1016/j.mbs.2018.02.004}
\end{barticle}
\endbibitem

\bibitem{Villaverde2018}
\begin{botherref}
\oauthor{\bsnm{Villaverde}, \binits{A.F.}}:
{Observability and structural identifiability of nonlinear biological systems}.
arXiv
\textbf{2019}
(2018)
{\href{https://arxiv.org/abs/1812.04525}{{arXiv:1812.04525}}}
\end{botherref}
\endbibitem

\bibitem{Fok2013}
\begin{barticle}
\bauthor{\bsnm{Fok}, \binits{P.-W.}},
\bauthor{\bsnm{Chou}, \binits{T.}}:
\batitle{Identifiability of age-dependent branching processes from extinction
  probabilities and number distributions}.
\bjtitle{Journal of Statistical Physics}
\bvolume{152}(\bissue{4}),
\bfpage{769}--\blpage{786}
(\byear{2013}).
\doiurl{10.1007/s10955-013-0781-1}
\end{barticle}
\endbibitem

\bibitem{laredo}
\begin{botherref}
\oauthor{\bsnm{Laredo}, \binits{C.}},
\oauthor{\bsnm{David}, \binits{O.}},
\oauthor{\bsnm{Garnier}, \binits{A.}}:
Inference for partially observed multitype branching processes and ecological
  applications
(2009).
\doiurl{10.48550/ARXIV.0902.4520}
\end{botherref}
\endbibitem

\bibitem{BalsaCanto2009AnII}
\begin{barticle}
\bauthor{\bsnm{Balsa-Canto}, \binits{E.}},
\bauthor{\bsnm{Alonso}, \binits{A.}},
\bauthor{\bsnm{Banga}, \binits{J.}}:
\batitle{An iterative identification procedure for dynamic modeling of
  biochemical networks}.
\bjtitle{BMC Systems Biology}
\bvolume{4},
\bfpage{11}--\blpage{11}
(\byear{2009})
\end{barticle}
\endbibitem

\bibitem{Balsa-Canto2008}
\begin{barticle}
\bauthor{\bsnm{Balsa-Canto}, \binits{E.}},
\bauthor{\bsnm{Alonso}, \binits{A.A.}},
\bauthor{\bsnm{Banga}, \binits{J.R.}}:
\batitle{{Computational procedures for optimal experimental design in
  biological systems}}.
\bjtitle{IET Syst. Biol.}
\bvolume{2}(\bissue{4}),
\bfpage{163}--\blpage{172}
(\byear{2008}).
\doiurl{10.1049/iet-syb:20070069}
\end{barticle}
\endbibitem

\bibitem{Srinath2010}
\begin{barticle}
\bauthor{\bsnm{Srinath}, \binits{S.}},
\bauthor{\bsnm{Gunawan}, \binits{R.}}:
\batitle{{Parameter identifiability of power-law biochemical system models}}.
\bjtitle{J. Biotechnol.}
\bvolume{149}(\bissue{3}),
\bfpage{132}--\blpage{140}
(\byear{2010}).
\doiurl{10.1016/j.jbiotec.2010.02.019}
\end{barticle}
\endbibitem

\bibitem{Browning2020}
\begin{barticle}
\bauthor{\bsnm{Browning}, \binits{A.P.}},
\bauthor{\bsnm{Warne}, \binits{D.J.}},
\bauthor{\bsnm{Burrage}, \binits{K.}},
\bauthor{\bsnm{Baker}, \binits{R.E.}},
\bauthor{\bsnm{Simpson}, \binits{M.J.}}:
\batitle{{Listen to the noise: identifiability analysis for stochastic
  differential equation models in systems biology}}.
\bjtitle{bioRxiv}
(\bissue{October})
(\byear{2020}).
\doiurl{10.1101/2020.08.10.245233}
\end{barticle}
\endbibitem

\bibitem{Gronwall1919}
\begin{barticle}
\bauthor{\bsnm{Gronwall}, \binits{T.H.}}:
\batitle{{Note on the Derivatives with Respect to a Parameter of the Solutions
  of a System of Differential Equations}}.
\bjtitle{Ann. Math.}
\bvolume{20}(\bissue{4}),
\bfpage{292}
(\byear{1919}).
\doiurl{10.2307/1967124}
\end{barticle}
\endbibitem

\bibitem{NYCHD2021}
\begin{botherref}
\oauthor{\bsnm{{NYC Health Department}}}:
{NYC Coronavirus Disease 2019 (COVID-19) Data}.
GitHub
(2021).
\url{https://github.com/nychealth/coronavirus-data}
\end{botherref}
\endbibitem

\bibitem{DEOLIVEIRA}
\begin{barticle}
\bauthor{\bsnm{{de Oliveira}}, \binits{A.C.S.}},
\bauthor{\bsnm{Morita}, \binits{L.H.M.}},
\bauthor{\bsnm{{da Silva}}, \binits{E.B.}},
\bauthor{\bsnm{Zardo}, \binits{L.A.R.}},
\bauthor{\bsnm{Fontes}, \binits{C.J.F.}},
\bauthor{\bsnm{Granzotto}, \binits{D.C.T.}}:
\batitle{Bayesian modeling of covid-19 cases with a correction to account for
  under-reported cases}.
\bjtitle{Infectious Disease Modelling}
\bvolume{5},
\bfpage{699}--\blpage{713}
(\byear{2020}).
\doiurl{10.1016/j.idm.2020.09.005}
\end{barticle}
\endbibitem

\bibitem{Richterich2020}
\begin{barticle}
\bauthor{\bsnm{Richterich}, \binits{P.}}:
\batitle{Severe underestimation of covid-19 case numbers: effect of epidemic
  growth rate and test restrictions}.
\bjtitle{medRxiv}
(\byear{2020})
{\href{https://arxiv.org/abs/https://www.medrxiv.org/content/early/2020/04/17/2020.04.13.20064220.full.pdf}{{https://www.medrxiv.org/content/early/2020/04/17/2020.04.13.20064220.full.pdf}}}.
\doiurl{10.1101/2020.04.13.20064220}
\end{barticle}
\endbibitem

\bibitem{LAU2021110}
\begin{barticle}
\bauthor{\bsnm{Lau}, \binits{H.}},
\bauthor{\bsnm{Khosrawipour}, \binits{T.}},
\bauthor{\bsnm{Kocbach}, \binits{P.}},
\bauthor{\bsnm{Ichii}, \binits{H.}},
\bauthor{\bsnm{Bania}, \binits{J.}},
\bauthor{\bsnm{Khosrawipour}, \binits{V.}}:
\batitle{Evaluating the massive underreporting and undertesting of covid-19
  cases in multiple global epicenters}.
\bjtitle{Pulmonology}
\bvolume{27}(\bissue{2}),
\bfpage{110}--\blpage{115}
(\byear{2021}).
\doiurl{10.1016/j.pulmoe.2020.05.015}
\end{barticle}
\endbibitem

\bibitem{Cori}
\begin{botherref}
\oauthor{\bsnm{Cori}, \binits{L.}},
\oauthor{\bsnm{Bianchi}, \binits{F.}},
\oauthor{\bsnm{Cadum}, \binits{E.}},
\oauthor{\bsnm{Anthonj}, \binits{C.}}:
Risk perception and covid-19.
International journal of environmental research and public health
\textbf{17}
(2020)
\end{botherref}
\endbibitem

\bibitem{Britton}
\begin{botherref}
\oauthor{\bsnm{Britton}, \binits{T.}},
\oauthor{\bsnm{Scalia~Tomba}, \binits{G.}}:
Estimation in emerging epidemics: biases and remedies.
J. R. Soc. Interface
\textbf{16}
(2019)
\end{botherref}
\endbibitem

\bibitem{Turkyilmazoglu}
\begin{barticle}
\bauthor{\bsnm{Turkyilmazoglu}, \binits{M.}}:
\batitle{{Explicit formulae for the peak time of an epidemic from the SIR
  model}}.
\bjtitle{Phys. D Nonlinear Phenom.}
\bvolume{422},
\bfpage{132902}
(\byear{2021}).
\doiurl{10.1016/j.physd.2021.132902}
\end{barticle}
\endbibitem

\bibitem{barlow}
\begin{barticle}
\bauthor{\bsnm{Barlow}, \binits{N.S.}},
\bauthor{\bsnm{Weinstein}, \binits{S.J.}}:
\batitle{Accurate closed-form solution of the sir epidemic model}.
\bjtitle{Physica D: Nonlinear Phenomena}
\bvolume{408},
\bfpage{132540}
(\byear{2020})
\end{barticle}
\endbibitem

\bibitem{schlickeiser}
\begin{barticle}
\bauthor{\bsnm{Schlickeiser}, \binits{R.}},
\bauthor{\bsnm{Kröger}, \binits{M.}}:
\batitle{Analytical solution of the {SIR}-model for the temporal evolution of
  epidemics: part b. semi-time case}.
\bjtitle{Journal of Physics A: Mathematical and Theoretical}
\bvolume{54}(\bissue{17}),
\bfpage{175601}
(\byear{2021}).
\doiurl{10.1088/1751-8121/abed66}
\end{barticle}
\endbibitem

\end{thebibliography}


\begin{appendices}



\section{Derivation of \edit{Approximation} \ref{prop1}}\label{sec:prop1}


The main observation leading to \eqref{eq:Perturbation} is that $s$ remains close to $1$ early in the epidemic. Motivated by this, we replace $s$ with $1$ in the SIR model to obtain
\begin{linenomath*}
\bel \label{approxODE}
	\frac{ds}{dt} = -\beta i,
	\quad
	\frac{di}{dt} = (\beta-\gamma)i = \delta i.
\eel
\end{linenomath*}
The corresponding solution starting from $x_0=(s_0,i_0)$ is
\begin{linenomath*}
\bel\label{approxsol}
	\widetilde\varphi_t(x_0,\theta) = \bigg(s_0 - \frac{\beta}{\delta}\big(e^{\delta t}-1\big)i_0,\ e^{\delta t}i_0\bigg).
\eel
\end{linenomath*}
Fix $\omega\in[0,2\pi)$ and set $\varphi_t= \varphi_t(x_0,\theta)$, $\varphi^\epsilon_t= \varphi_t(x_0,\theta_\epsilon(\omega))$, $\widetilde{\varphi}_t=\widetilde{\varphi}_t(x_0,\theta)$, and $\widetilde{\varphi}^\epsilon_t= \widetilde{\varphi}_t(x_0,\theta_\epsilon(\omega))$. The expression of interest, $\lVert\varphi^\epsilon_t-\varphi_t\rVert$, can be written
\begin{linenomath*}
\bel\label{eq:error}
\begin{aligned}
	\lVert\varphi^\epsilon_t-\varphi_t\rVert &= \lVert\widetilde\varphi^\epsilon_t-\widetilde\varphi_t\rVert+E^\epsilon_t
\end{aligned}
\eel
\end{linenomath*}
where $E^\epsilon_t$ is the error incurred from approximating $\lVert\varphi^\epsilon_t-\varphi_t\rVert$ by $\lVert\widetilde\varphi^\epsilon_t-\widetilde\varphi_t\rVert$. As mentioned in the main text, we do not have explicit analytic control over $E^\epsilon_t$ but the numerical analysis in Appendix \ref{sec:error} indicates it is negligible compared to $\lVert\widetilde\varphi^\epsilon_t-\widetilde\varphi_t\rVert$ early in the epidemic. In particular, the approximation 
\begin{linenomath*}
\bel\label{normapprox}
	\lVert \widetilde\varphi^\epsilon_t-\widetilde\varphi_t\rVert
		\approx \lVert \varphi^\epsilon_t-\varphi_t\rVert
\eel
\end{linenomath*}
is valid up to roughly 80\% of the time of peak infection. Thus we turn attention to $\widetilde\varphi^\epsilon_t-\widetilde\varphi_t$. Fixing $t$ and Taylor expanding $\widetilde\varphi_t$ to first order in $\eta= (\beta,\delta)$ gives
\begin{linenomath*}
\be
\begin{aligned}
	\widetilde\varphi^\epsilon_t-\widetilde\varphi_t &= D_{\eta}\widetilde\varphi_t(x,\eta)\big(\eta_\epsilon(\omega)-\eta\big) + o\big(\lVert\eta\rVert^2\big) \\
		&\approx
		\epsilon\begin{pmatrix}
			-\tfrac{1}{\delta}\big(e^{\delta t}-1\big)i_0
			& 
			-\tfrac{\beta}{\delta}\bigg(te^{\delta t}-\tfrac{1}{\delta}\big(e^{\delta t}-1\big)\bigg)i_0 \\
			0
			&
			te^{\delta t}i_0
		\end{pmatrix}
		\begin{pmatrix} \cos(\omega) \\ \cos(\omega)-\sin(\omega)\end{pmatrix} \\
		&\approx
		\epsilon i_0\begin{pmatrix}
			-t
			& 
			-\beta t^2 \\
			\phantom{-}0
			&
			\phantom{-}te^{\delta t}
		\end{pmatrix}
		\begin{pmatrix} \cos(\omega) \\ \cos(\omega)-\sin(\omega)\end{pmatrix}
\end{aligned}
\ee
\end{linenomath*}
where $D_{\eta}$ is the derivative in $\eta$ and $\eta_\epsilon(\omega)=\eta+\epsilon(\cos(\omega)-\sin(\omega))$. The first approximation is a result of simply dropping the $o(\lVert\eta\rVert^2)$ term and the second is obtained by substituting the first order Taylor approximation $\exp(\delta t)-1=\delta t+o(\delta^2)$ about $\delta=0$ into the expression above it. We see from the latter expression that any vector $(x_1,x_2)$ in $\mathbb{R}^2$ with $x_2\neq 0$ will grow exponentially in time under the above matrix due to the $te^{\delta t}$ term in the bottom right. On the other hand, the first component $x_1$ will only grow linearly in time provided $x_2=0$. So the magnitude of growth is minimized for vectors of the form $(x_1,0)$. This implies that, under the above approximations which come at a cost of $o(\lVert\eta\rVert^2)$ and $o(\delta^2)$, respectively, the difference $\tilde\varphi^\epsilon_t-\tilde\varphi_t$ will grow the least when $(\cos(\omega), \cos(\omega)-\sin(\omega))=\pm (1,0)$. Therefore the perturbations $\eta_\epsilon(\omega)$ that yield the smallest separation between $\widetilde\varphi^\epsilon_t$ and $\widetilde\varphi_t$ are those corresponding to the directions $\pm(1,0)$ or, equivalently, the angles $\omega=\pi/4$ and $5\pi/4$. Now
\begin{linenomath*}
\bel\label{eq7}
	\widetilde{\varphi}^\epsilon_t(x_0,\theta_\epsilon(\omega))-\widetilde{\varphi}_t(x_0,\theta) =
	\begin{pmatrix}
		\tfrac{\beta_\epsilon}{\delta_\epsilon}-\tfrac{\beta}{\delta}+\big(\tfrac{\beta}{\delta}-\tfrac{\beta_\epsilon}{\delta_\epsilon}e^{\epsilon t(\cos\omega-\sin\omega)}\big)e^{\delta t} \\
		\big(e^{\epsilon t(\cos\omega-\sin\omega)}-1\big)e^{\delta t}
	\end{pmatrix}
	i_0
\eel
\end{linenomath*}
where $\beta_\epsilon= \beta+\epsilon\cos\omega$ and $\delta_\epsilon= \delta+\epsilon(\cos\omega-\sin\omega)$. Plugging in $\pi/4$ gives
\begin{linenomath*}
\bel\label{eq:13}
\begin{aligned}
	\widetilde{\varphi}^\epsilon_t(x_0,\theta_\epsilon(\omega_1))-\widetilde{\varphi}_t(x_0,\theta) &=
	\begin{pmatrix}
		1-e^{\delta t} \\ 0
	\end{pmatrix}
	\frac{\epsilon i_0}{\delta\sqrt{2}}
\end{aligned}
\eel
\end{linenomath*}
and similarly for $5\pi/4$, only negative. Thus, in combination with \eqref{normapprox},
\begin{linenomath*}
\be
\begin{aligned}
	\frac{\epsilon}{\delta\sqrt{2}}\big(e^{\delta t}-1\big)i_0 &\lesssim \lVert \widetilde\varphi^\epsilon_t-\widetilde\varphi_t\rVert
		\approx \lVert \varphi^\epsilon_t-\varphi_t\rVert
\end{aligned}
\ee
\end{linenomath*}
for all $\omega$, which is precisely \eqref{eq:Perturbation}.


\section{Numerical analysis of error}\label{sec:error}


\begin{figure}[!ht]
\centering
\includegraphics[width=1\textwidth]{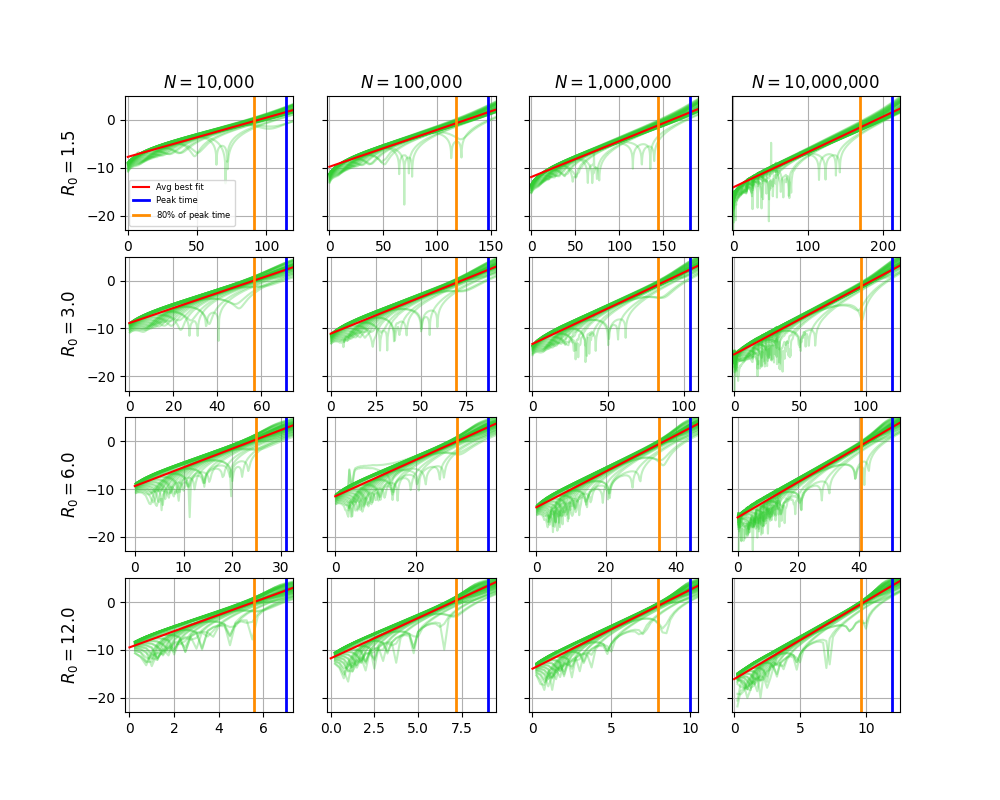}
\caption{\textit{Error analysis corresponding to Figure \ref{fig2}}. As in Figure \ref{fig2}, the horizontal axis in each graph is the number of days since the start of the epidemic, the vertical blue line is the peak time, and the orange line is 80\% of the peak time. The vertical axis is the logarithm of the relative error, $\log\lvert E^\epsilon_t\rvert-\log\lVert\varphi^\epsilon_t-\varphi_t\rVert$, at time $t$. There are 50 green curves in each plot which are the relative errors corresponding to 50 angles in the intervals $[\pi/4-\pi/12,\pi/4+\pi/12)$ and $[5\pi/4-\pi/12,5\pi/4+\pi/12)$. The red line in each plot is the average linear approximation of the 50 green curves. Also as in Figure \ref{fig2}, the first through fourth columns have population sizes $10^4, 10^5, 10^6$, and $10^7$, respectively with one initial infection in each case, and the first through fourth rows have parameters $(\beta,\gamma,\epsilon)=(.21,.14,.03), (.21,.07,.03), (.42,.07,.06)$, and $(1.68,.14,.1)$ which give $R_0$ values of 1.5, 3, 6, and 12. The main point is that for every combination of parameters the relative error is exponentially small until about 80\% of the peak time at which point it becomes $\mathcal{O}(1)$ and subsequently blows up exponentially.} \label{fig:logerror}
\end{figure}

In this section we revisit the expression
\begin{linenomath*}
\be
\begin{aligned}
	\lVert\varphi-\varphi_t\rVert &= \lVert\widetilde\varphi^\epsilon_t-\widetilde\varphi_t\rVert+E^\epsilon_t
\end{aligned}
\ee
\end{linenomath*}
which is \eqref{eq:error} in Appendix \ref{sec:prop1}; see the sentences preceding \eqref{eq:error} for notation. A key requirement in the derivation of \edit{Approximation} \ref{prop1} is that $\lVert\varphi_t^\epsilon-\varphi_t\rVert$ is well-approximated by $\lVert\widetilde\varphi_t^\epsilon-\widetilde\varphi_t\rVert$. By \eqref{eq:error} this is the case whenever
\begin{linenomath*}
\bel\label{eq:small}
\begin{aligned}
	\frac{\lvert E^\epsilon_t\rvert}{\lVert\varphi^\epsilon_t-\varphi_t\rVert} < \tau
\end{aligned}
\eel
\end{linenomath*}
for some prescribed tolerance $\tau$, which we take to be 1 both for simplicity and because it agrees with the numerical observations from Figure \ref{fig2}. As mentioned before, we do not have sufficient control to verify \eqref{eq:small} analytically. However, we can verify \eqref{eq:small} numerically -- in our case with the odeint function from the scipy.integrate package in Python. Specifically, we numerically solve both the SIR and approximate SIR equations and compute
\begin{linenomath*}
\bel\label{eq:logsmall}
\begin{aligned}
	\log\lvert E^\epsilon_t\rvert-\log\lVert\varphi^\epsilon_t-\varphi_t\rVert
\end{aligned}
\eel
\end{linenomath*}
for 25 evenly spaced angles $\omega$ in the interval $[\pi/4-\pi/12,\pi/4+\pi/12)$ and 25 more such angles in $[5\pi/4-\pi/12,5\pi/4+\pi/12)$. We restrict attention to these intervals because, as observed in Figure \ref{fig2}, these are the angles for which the inverse problem is least practically identifiable. The results, plotted in Figure \ref{fig:logerror}, indicate \eqref{eq:logsmall} grows approximately linearly in time. For each of the 16 subplots in Figure \ref{fig:logerror} (which correspond to 16 different $\beta, \gamma, \epsilon$, and $N$ combinations) we average the linear approximations for each of the 50 curves to obtain a single ``average" linear approximation indicated by the red line in each subplot. Table \ref{tab:logerror} gives the equations for each of these averaged lines. So, for example, when $(\beta,\gamma)=(.21,.07)$, $\epsilon=.03$, and $N=10^6$ the table gives
\begin{linenomath*}
\be
\begin{aligned}
	\log\lvert E^\epsilon_t\rvert-\log\lVert\varphi^\epsilon_t-\varphi_t\rVert \approx 0.15t - 13.3.
\end{aligned}
\ee
\end{linenomath*}
Setting this expression equal to 0 (which corresponds to $\tau=1$) and solving for $t$ gives $t=13.3/0.15\approx 89$ days, i.e. the error becomes intolerable at roughly 89 days. We can then divide 89 by the peak time to estimate the percentage of the peak time at which the error becomes sufficiently large that the approximation, and hence the lower bound from \edit{Approximation} \ref{prop1}, no longer holds. The resulting percentages are given in Table \ref{tab:percentpeak} and motivate our choice of 80\% in \edit{Approximation} \ref{prop1}.

\begin{table}[!ht]
\centering
\begin{tabular}{|c|c|c|c|c|}
	\hline
	& & & & \\
	$\boldsymbol{(\beta,\gamma,\epsilon)}$ & $\boldsymbol{N=10^4}$ & $\boldsymbol{N=10^5}$ &  $\boldsymbol{N=10^6}$ & $\boldsymbol{N=10^7}$ \\
	\hline
	& & & & \\
	$(.21,.14,.03)$ & $.08t-7.7$ & $.08t-9.8$ & $.08t-11.9$ & $.07t-14.1$ \\
	\hline
	& & & & \\
	$(.21,.07,.03)$ & $.16t-8.9$ & $.15t-11.1$ & $.15t-13.3$ & $.15t-15.4$ \\
	\hline
	& & & & \\
	$(.42,.07,.06)$ & $.39t-9.4$ & $.38t-11.5$ & $.38t-13.8$ & $.37t-15.9$ \\
	\hline
	& & & & \\
	$(1.68,.14,.1)$ & $1.7t-9.5$ & $1.68t-11.8$ & $1.65t-14$ & $1.63t-16.1$ \\
	\hline
\end{tabular}
\vspace{\baselineskip}
\caption{\textit{Linear approximations of log relative error}. Equations for the average linear approximations to the log relative error $\log\lvert E^\epsilon_t\rvert-\log\lVert\varphi^\epsilon_t-\varphi_t\rVert$ as functions of time for each combination of parameter values appearing in Figure \ref{fig2} and \ref{fig:logerror}.}
\label{tab:logerror}
\end{table}

\begin{table}[!ht]
\centering
\begin{tabular}{|c|c|c|c|c|}
	\hline
	& & & & \\
	$\boldsymbol{(\beta,\gamma,\epsilon)}$ & $\boldsymbol{N=10^4}$ & $\boldsymbol{N=10^5}$ &  $\boldsymbol{N=10^6}$ & $\boldsymbol{N=10^7}$ \\
	\hline
	& & & & \\
	$(.21,.14,.03)$ & $83\%$ & $86\%$ & $88\%$ & $90\%$ \\
	\hline
	& & & & \\
	$(.21,.07,.03)$ & $80\%$ & $83\%$ & $85\%$ & $87\%$ \\
	\hline
	& & & & \\
	$(.42,.07,.06)$ & $77\%$ & $80\%$ & $83\%$ & $84\%$ \\
	\hline
	& & & & \\
	$(1.68,.14,.1)$ & $80\%$ & $78\%$ & $85\%$ & $83\%$ \\
	\hline
\end{tabular}
\vspace{\baselineskip}
\caption{\textit{Percent of peak time at which error becomes intolerable}. Approximate percent of peak time when the log relative error $\log\lvert E^\epsilon_t\rvert-\log\lVert\varphi^\epsilon_t-\varphi_t\rVert$ becomes intolerable, i.e. $\mathcal{O}(1)$.}
\label{tab:percentpeak}
\end{table}

\newpage

\edit{

\section{Theoretical analysis of error}\label{sec:theory}


In this section we again revisit \eqref{eq:error} and prove the following.

\begin{prop}\label{prop:theory}
For every initial condition $x_0=(s_0,i_0)$, parameter pair $(\beta,\gamma)$, $\epsilon>0$, $\omega$ in $[0,2\pi)$, and $t\geq 0$,
\bel\label{eq:theory}
	\lvert E^\epsilon_t\rvert \leq \left( \frac{\sqrt{2\beta_\epsilon^2+\gamma_\epsilon^2}}{\delta_\epsilon}\left(e^{\delta_\epsilon t}-1\right) + \frac{\sqrt{2\beta^2+\gamma^2}}{\delta}\left(e^{\delta t}-1\right)\right)i_0,
\eel
where, as before, $\delta=\beta-\gamma$, $\beta_\epsilon=\beta+\epsilon\cos(\omega)$, $\gamma_\epsilon=\gamma+\epsilon\sin(\omega)$, and $\delta_\epsilon=\beta_\epsilon-\gamma_\epsilon$.
\end{prop}

\begin{proof}[\textit{Proof}.]
By \eqref{eq:error}, the reverse triangle inequality, and the triangle inequality,
\bel\label{eq:proof1}
\begin{aligned}
	\lvert E^\epsilon_t\rvert &= \left\lvert \lVert \varphi^\epsilon_t-\varphi_t\rVert - \lVert \widetilde\varphi^\epsilon_t-\widetilde\varphi_t\rVert\right\rvert
		\leq \lVert \varphi^\epsilon_t-\varphi_t-\widetilde\varphi^\epsilon_t-\widetilde\varphi_t\rVert \\
		&\leq \lVert \widetilde\varphi^\epsilon_t-\varphi^\epsilon_t\rVert + \lVert\widetilde\varphi_t - \varphi_t\rVert.
\end{aligned}
\eel
Next, note that $\widetilde\varphi_t-\varphi_t = \int_0^t \frac{d}{d\tau}(\widetilde\varphi_\tau-\varphi_\tau)d\tau$ and hence
\be
\begin{aligned}
	\lVert\widetilde\varphi_t-\varphi_t\rVert &\leq \int_0^t \left\lVert\tfrac{d}{d\tau}(\widetilde\varphi_\tau-\varphi_\tau)\right\rVert d\tau \\
		&= \int_0^t \left\lVert \left(-\beta(\widetilde \iota_\tau-i_\tau s_\tau), \beta(\widetilde \iota_\tau-i_\tau s_\tau)-\gamma(\widetilde \iota_\tau-i_\tau)\right)\right\rVert d\tau \\
		&= \int_0^t \sqrt{2\beta^2(\widetilde\iota_\tau-i_\tau s_\tau)^2 - 2\beta\gamma(\widetilde\iota_\tau-i_\tau s_\tau)(\widetilde\iota_\tau - i\tau) + \gamma^2(\widetilde\iota_\tau - i_\tau)^2} d\tau \\
		&\leq \int_0^t \widetilde\iota_\tau\sqrt{2\beta^2+\gamma^2} d\tau
		= \frac{\sqrt{2\beta^2+\gamma^2}}{\delta}\left(e^{\delta t}-1\right)i_0.
\end{aligned}
\ee
The first equality holds because $\widetilde\varphi$ and $\varphi$ are solutions of \eqref{approxODE} and \eqref{sir}, respectively, and the second inequality holds because $\widetilde\iota_t\geq i_t\geq i_ts_t$ for all $t\geq 0$ and hence
\bel\label{eq:proof2}
	2\beta^2(\widetilde\iota_\tau-i_\tau s_\tau)^2 - 2\beta\gamma(\widetilde\iota_\tau-i_\tau s_\tau)(\widetilde\iota_\tau - i\tau) + \gamma^2(\widetilde\iota_\tau - i_\tau)^2 \leq (2\beta^2+\gamma^2)\widetilde\iota_\tau^2
\eel
for every $\tau\geq 0$. Finally, the last equality follows from $\widetilde\iota_\tau = e^{\delta\tau}i_0$. Since the above holds for arbitrary $\beta$ and $\gamma$, we also have
\be
	\lVert\widetilde\varphi^\epsilon_t-\varphi^\epsilon_t\rVert \leq \frac{\sqrt{2\beta_\epsilon^2+\gamma_\epsilon^2}}{\delta_\epsilon}\left(e^{\delta_\epsilon t}-1\right)i_0.
\ee
Combining with \eqref{eq:proof1} therefore gives
\begin{equation*}
	\lvert E^\epsilon_t\rvert \leq \left( \frac{\sqrt{2\beta_\epsilon^2+\gamma_\epsilon^2}}{\delta_\epsilon}\left(e^{\delta_\epsilon t}-1\right) + \frac{\sqrt{2\beta^2+\gamma^2}}{\delta}\left(e^{\delta t}-1\right)\right)i_0,
\end{equation*}
as claimed.
\end{proof}

A few remarks are in order. First, if $i_0=c/N$ for some constant $c$, then $i_0$, and hence $E^\epsilon_t$ for any finite $t$, go to $0$ as the population size $N$ goes to infinity. Thus Approximation \ref{prop1} improves with increasing $N$, which agrees with the numerical results of the last section. Second, the bound \eqref{eq:theory} is not sharp. This is due to the use of the triangle and reverse triangle inequalities, as well as the other inequalities implemented in the above proof. In particular, inequality \eqref{eq:proof2} is quite coarse because we have essentially thrown away all $i_\tau$ and $s_\tau$ terms. This is rather necessary since the SIR model has no known exact analytic solution, so gaining explicit control over terms involving $i_\tau$ and $s_\tau$ is, to the best of our knowledge, largely intractable \cite{barlow}. Finally, the upper bound \eqref{eq:theory} grows exponentially in time. This is a byproduct of the fact that the proportion of infected individuals $\widetilde\iota$ in the approximate dynamics \eqref{approxODE} grows exponentially without bound. On the other hand, growth of the $i$ compartment in the true SIR dynamics \eqref{sir} is offset by the decreasing number of susceptible individuals as the epidemic progresses, i.e. fewer susceptible individuals means there are fewer people to infect. This is the primary reason why the approximate dynamics are only valid early in the epidemic, and certainly do not hold beyond the time of peak infection.}



\section{Derivation of \edit{Approximation} \ref{prop2}}


Fix $\omega\in[0,2\pi)$ and set $\theta_\epsilon= \theta_\epsilon(\omega)$, $\Delta_t^\epsilon=\Delta_t(\theta_\epsilon)$, and $\Delta_t^0= \Delta_t(\theta_0)$. The likelihood for observed data $Y_{1:T}$ is
\begin{linenomath*}
\be
	L(Y_{1:T}\vert\theta_0) = \prod_{t=1}^T \frac{1}{\sqrt{2\pi\sigma_t^2}}\exp\bigg(-\frac{1}{2\sigma_t^2}\big(Y_t-p\Delta_t^0\big)^2\bigg)
\ee
\end{linenomath*}
and similarly for $L(Y_{1:T}\vert\theta_\epsilon)$. So the log-likelihood ratio between $\theta_\epsilon$ and $\theta_0$ is
\begin{linenomath*}
\be
\begin{aligned}
	\log\bigg(\frac{L(Y_{1:T}\vert\theta_\epsilon)}{L(Y_{1:T}\vert\theta_0)}\bigg) &= \sum_{t=1}^T \frac{1}{2\sigma_t^2}\bigg[\big(Y_t-p\Delta_t^0\big)^2-\big(Y_t-p\Delta_t^\epsilon\big)^2\bigg] \\
		&= \sum_{t=1}^T \frac{1}{2\sigma_t^2}\bigg[2pY_t\big(\Delta_t^\epsilon-\Delta_t^0\big)-p^2\big((\Delta_t^\epsilon)^2-(\Delta_t^0)^2\big)\bigg].
\end{aligned}
\ee
\end{linenomath*}
If the $Y_t$ satisfy \eqref{eq:NoisySIR} with parameter $\theta_\epsilon$, i.e. $Y_t=p\Delta_t^\epsilon+\xi_t$, then
\begin{linenomath*}
\be
	2pY_t\big(\Delta_t^\epsilon-\Delta_t^0\big)-p^2\big((\Delta_t^\epsilon)^2-(\Delta_t^0)^2\big) = p^2\big(\Delta_t^\epsilon-\Delta_t^0\big)^2 +2p\big(\Delta_t^\epsilon-\Delta_t^0\big)\xi_t.
\ee
\end{linenomath*}
So, letting $\eta$ be the value in \eqref{eq:LRT} for the likelihood ratio test,
\begin{linenomath*}
\be
\begin{aligned}
	\mathcal{E}_2(\omega) &= \mathbb{P}_{\theta_\epsilon}\bigg(\frac{L(Y_{1:T}\vert\theta_\epsilon)}{L(Y_{1:T}\vert\theta_0)} < \eta\bigg)
		= \mathbb{P}_{\theta_\epsilon}\bigg(\log\bigg(\frac{L(Y_{1:T}\vert\theta_\epsilon)}{L(Y_{1:T}\vert\theta_0)}\bigg) < \log\eta\bigg) \\
		&= \mathbb{P}_{\theta_\epsilon}\bigg(\sum_{t=1}^T \frac{1}{2\sigma_t^2}\bigg[p^2\big(\Delta_t^\epsilon-\Delta_t^0\big)^2 +2p\big(\Delta_t^\epsilon-\Delta_t^0\big)\xi_t\bigg] < \log\eta\bigg) \\
		&= \mathbb{P}_{\theta_\epsilon}\bigg(\sum_{t=1}^T \frac{p}{\sigma_t^2}\big(\Delta_t^\epsilon-\Delta_t^0\big)\xi_t < \log\eta-\frac{1}{2}V_T^\epsilon(\omega)\bigg).
\end{aligned}
\ee
\end{linenomath*}
where 
\begin{linenomath*}
\bel\label{eq:V_T}
	V_T^\epsilon(\omega) = \sum_{t=1}^T \frac{p^2}{\sigma_t^2}(\Delta_t^\epsilon-\Delta_t^0)^2. 
\eel
\end{linenomath*}
Now $\xi_t\sim\mathcal{N}(0,\sigma_t^2)$ implies $\sum_{t=1}^T(p/\sigma_t^2)(\Delta_t^\epsilon-\Delta_t^0)\xi_t\sim\mathcal{N}(0,V^\epsilon_T)$ and hence
\begin{linenomath*}
\be
	\frac{1}{\sqrt{V_T^\epsilon(\omega)}}\sum_{t=1}^T \frac{p}{\sigma^2_t}(\Delta^\epsilon_t-\Delta^0_t)\xi_t \sim \mathcal{N}(0, 1).
\ee
\end{linenomath*}
So, letting $Z$ denote a standard normal random variable,
\begin{linenomath*}
\be
	\mathcal{E}_2(\omega) = \mathbb{P}_{\theta_\epsilon}\bigg( Z < \frac{\log\eta}{\sqrt{V_T^\epsilon(\omega)}}-\frac{1}{2}\sqrt{V_T^\epsilon(\omega)}\bigg)
		= \Phi\bigg(\frac{\log \eta}{\sqrt{V_T^\epsilon(\omega)}}-\frac{1}{2}\sqrt{V_T^\epsilon(\omega)}\bigg).
\ee
\end{linenomath*}
By an entirely similar computation (note the symmetry between $\theta_\epsilon$ and $\theta_0$),
\begin{linenomath*}
\be
	\alpha = \mathbb{P}_{\theta_0}\bigg(\frac{L(Y_{1:T}\vert\theta_\epsilon)}{L(Y_{1:T}\vert\theta_0)} \geq \eta\bigg)
		= \Phi\bigg(-\frac{\log \eta}{\sqrt{V_T^\epsilon(\omega)}}-\frac{1}{2}\sqrt{V_T^\epsilon(\omega)}\bigg)
\ee
\end{linenomath*}
and so $\frac{\log\eta}{\sqrt{V_T^\epsilon(\omega)}} = -\Phi^{-1}(\alpha)-\frac{1}{2}\sqrt{V_T^\epsilon(\omega)}$. Therefore, since $\Phi(-x)=1-\Phi(x)$,
\begin{linenomath*}
\bel\label{eq:171}
	\mathcal{E}_2(\omega) = \Phi\bigg(-\Phi^{-1}(\alpha)-\sqrt{V_T^\epsilon(\omega)}\bigg)
		= 1 - \Phi\bigg(\Phi^{-1}(\alpha)+\sqrt{V_T^\epsilon(\omega)}\bigg).
\eel
\end{linenomath*}
To obtain the approximations, first note that from \eqref{approxsol} we have
\begin{linenomath*}
\bel\label{eq:9}
	\Delta^\epsilon_t(\omega)-\Delta^0_t \approx \left[\beta_\epsilon\bigg(\frac{e^{-\delta_\epsilon}-1}{-\delta_\epsilon}\bigg)e^{\epsilon(\cos\omega-\sin\omega)t}-\beta\bigg(\frac{e^{-\delta}-1}{-\delta}\bigg)\right]Ni_0e^{\delta t},
\eel
\end{linenomath*}
where $\beta_\epsilon=\beta+\epsilon\cos\omega$ and $\delta_\epsilon=\delta+\epsilon(\cos\omega-\sin\omega)$. This can be further simplified by using the Taylor approximation $(e^x-1)/x\approx 1$ to obtain
\begin{linenomath*}
\bel\label{eq:181}
	\Delta^\epsilon_t(\omega)-\Delta^0_t \approx \left[\beta_\epsilon e^{\epsilon(\cos\omega-\sin\omega)t}-\beta\right]Ni_0e^{\delta t}.
\eel
\end{linenomath*}
Plugging \eqref{eq:9} and subsequently \eqref{eq:181} into \eqref{eq:171} then gives
\begin{linenomath*}
\footnotesize
\begin{align*}
	\mathcal{E}_2(
	\omega) &\approx  1-\Phi\left(\Phi^{-1}(\alpha)+pNi_0\sqrt{\sum_{t=1}^T\frac{e^{2\delta t}}{\sigma_t^2}\left[\beta_\epsilon\bigg(\frac{e^{-\delta_\epsilon}-1}{-\delta_\epsilon}\bigg)e^{\epsilon(\cos\omega-\sin\omega)t}-\beta\bigg(\frac{e^{-\delta}-1}{-\delta}\bigg)\right]^2}\phantom{-}\right) \\
	    &\approx 1-\Phi\left(\Phi^{-1}(\alpha)+pNi_0\sqrt{\sum_{t=1}^T\frac{e^{2\delta t}}{\sigma_t^2}\left[(\beta+\epsilon\cos\omega) e^{\epsilon(\cos\omega-\sin\omega)t}-\beta\right]^2}\phantom{-}\right),
\end{align*}
\end{linenomath*}
\normalsize
as claimed. Finally, from \edit{Approximation} \ref{prop1} we know the directions of least separation between the trajectories corresponding to $\theta_0$ and $\theta_\epsilon(\omega)$ are approximately $\omega=\pi/4$ and $5\pi/4$. This suggests it will be most difficult to distinguish the null and alternative hypotheses of \eqref{eq:hypothesisSIR} when $\omega=\pi/4$ or $5\pi/4$. So when constrained to a significance level $\alpha$, one of these two angles will approximately maximize the type II error rate $\mathcal{E}_2(\omega)$ over all $\omega\in[0,2\pi)$. 

\end{appendices}

\end{document}